\documentclass[12pt]{article}

\usepackage[utf8]{inputenc}

\usepackage{csquotes}

\usepackage[style=alphabetic,
  sorting = nyt,
  sortcites=true,
  giveninits=true,
  date=year,
  isbn=false,
  maxbibnames=99,
  maxalphanames=6, 
  backend=bibtex]{biblatex}   
\DeclareFieldFormat*{titlecase}{\MakeSentenceCase{#1}}

\DeclareSourcemap{\maps[datatype=bibtex]{\map{\step[fieldset=shorthand, null]}}}

\DeclareSourcemap{
  \maps[datatype=bibtex]{
    \map[overwrite]{
      \step[fieldsource=doi, final]
      \step[fieldset=url, null]
      \step[fieldset=eprint, null]
    }  
  }
}

\DeclareSourcemap{
  \maps[datatype=bibtex, overwrite]{
    \map{
      \step[fieldset=language, null]
      \step[fieldset=month, null]
      \step[fieldset=urldate, null]
    }
  }
}

\AtEveryBibitem{%
  \clearlist{language}%
  \clearlist{month}%
  \clearlist{urldate}%
  \clearfield{urldate}%
}

\renewbibmacro{in:}{}

\DeclareFieldFormat[misc]{title}{\mkbibquote{#1}}
\DeclareFieldFormat[article,inproceedings]{volume}{\mkbibbold{#1}}
\DeclareFieldFormat[inproceedings]{series}{\mkbibitalic{#1}}

\usepackage[english]{babel}
 \usepackage{amsmath}
\usepackage{amsthm}
\usepackage{amssymb}
\usepackage{hyperref}
\usepackage{braket}
\usepackage{color}
\usepackage{comment}
\usepackage{graphicx}
\usepackage[margin=2.5cm]{geometry}
\usepackage[noblocks]{authblk}
\usepackage[noabbrev]{cleveref}

\usepackage{enumerate}

\usepackage{chngcntr}
\counterwithin{figure}{section}
\counterwithin{equation}{section}
\counterwithout{figure}{section}

\usepackage[UKenglish]{isodate}
\cleanlookdateon

\let\C\relax

\newcommand{\epsi}{\varepsilon}
\newcommand{\eps}{\varepsilon}
\newcommand{\E}{{\mathrm{e}}}
\newcommand{\I}{\mathrm{i}}
 \newcommand{\R}{ \mathbb{R} }
  \newcommand{\Sph}{ \mathbb{S} }
\renewcommand{\S}{\mathbb{S}}
\newcommand{\C}{ \mathbb{C} }
\newcommand{\N}{ \mathbb{N} }

\newcommand{\D}{\mathrm{d}}

\newcommand{\mcE}{\mathcal{E}}

\newcommand{\mfF}{\mathfrak{F}}

\newcommand{\mcV}{\mathcal{V}}

\newcommand{\fBCS}{\mathsf{f}_{\textnormal{BCS}}}

 \newcommand{\norm}[1]{\left\Vert #1 \right\Vert} 
 \newcommand{\abs}[1]{\left\vert #1 \right\vert}
\newcommand{\longip}[3]{\left\langle #1 \middle\vert #2 \middle\vert #3 \right\rangle}

\DeclareMathOperator{\sgn}{sgn}
\DeclareMathOperator{\spec}{spec}

\DeclareMathOperator{\mathspan}{span}

\newcommand{\ud}{\,\textnormal{d}}

\newcommand{\ee}{\mathrm{e}}

\usepackage{seqsplit}

\theoremstyle{plain}
\newtheorem{thm}{Theorem}[section]
\newtheorem{lem}[thm]{Lemma}

\newtheorem{prop}[thm]{Proposition}

\newtheorem{lemma}[thm]{Lemma}

\theoremstyle{definition}

\newtheorem{rmk}[thm]{Remark}
\newtheorem{assumption}[thm]{Assumption}

\newtheorem{remark}[thm]{Remark}

\crefname{thm}{theorem}{theorems}
\crefname{problem}{problem}{problems}
\crefname{lemma}{lemma}{lemmas}
\crefname{lem}{lemma}{lemmas}
\crefname{cor}{corollary}{corollaries}
\crefname{prop}{proposition}{propositions}
\crefname{conj}{conjecture}{conjectures}
\crefname{defn}{definition}{definitions}
\crefname{defi}{definition}{definitions}
\crefname{note}{note}{notes}
\crefname{ex}{example}{examples}
\crefname{remark}{remark}{remarks}
\crefname{rmk}{remark}{remarks}
\crefname{notation}{notation}{notations}
\crefname{assumption}{assumption}{assumptions}
\crefname{claim}{claim}{claims}
\crefname{claim*}{claim}{claims}

\allowdisplaybreaks

\title{Universal behavior of the BCS energy gap}
\author{Joscha Henheik\footnote{\href{mailto:joscha.henheik@ist.ac.at}{joscha.henheik@ist.ac.at}}~~and Asbjørn Bækgaard Lauritsen\footnote{\href{mailto:alaurits@ist.ac.at}{alaurits@ist.ac.at}} \\ IST Austria, Am Campus 1, 3400 Klosterneuburg, Austria}

\begin{document}
\maketitle
\begin{abstract}
	We consider the BCS 
 energy gap $\Xi(T)$ (essentially given by $\Xi(T) \approx \Delta(T, \sqrt\mu)$, the BCS order parameter) 
  at all temperatures $0 \le T \le T_c$ up to the critical one, $T_c$, 
  and show that, in the limit of weak coupling, the ratio 
   $\Xi(T)/T_c$ 
  is given by a universal function of the relative temperature $T/T_c$. On the one hand, this recovers a recent result by Langmann and Triola 
  (Phys. Rev. B 108.10 (2023))
  on three-dimensional $s$-wave superconductors for temperatures bounded uniformly away from $T_c$. On the other hand, our result lifts these restrictions, as we consider arbitrary 
	spatial dimensions $d \in \{1,2,3\}$, discuss superconductors with non-zero angular momentum (primarily in two dimensions), and treat the perhaps physically most interesting (due to the occurrence of the superconducting phase transition) regime of temperatures close to $T_c$.  
~
\end{abstract}

\section{Introduction} \label{sec:intro}
The Bardeen--Cooper--Schrieffer (BCS) \cite{bcs.original} theory of superconductivity exhibits many interesting features.
Notably it predicts, for $s$-wave superconductors (i.e.~those with angular momentum $\ell = 0$ and a radially symmetric gap function), 
that the superconducting energy gap $\Xi$ {(essentially given by $\Xi \approx \Delta(\sqrt\mu)$, the BCS order parameter) }is proportional to the critical temperature $T_c$
with a universal proportionality constant independent of the microscopic details of the electronic interactions, i.e.~the specific superconductor. 
At zero temperature, the claimed universality is the (approximate) formula $\Xi/T_c \approx \pi e^{-\gamma} \approx 1.76$ with $\gamma \approx 0.57$ the 
Euler--Mascheroni constant, a property which is well-known in the physics literature \cite{nozieres.schmitt-rink,bcs.original}. More recently, based on the variational formulation of BCS theory, mostly developed by Hainzl and Seiringer with others \cite{hainzl.hamza.seiringer.solovej, frank.hainzl.naboko.seiringer, hainzl.seiringer.review.08, hainzl.seiringer.16}, it has been put on rigorous grounds in various (physically quite different) limiting regimes 
\cite{frank.hainzl.naboko.seiringer,Hainzl.Seiringer.2008,hainzl.seiringer.scat.length,lauritsen.energy.gap.2021,Henheik.2022,Henheik.Lauritsen.2022,Henheik.Lauritsen.ea.2023} (see Section \ref{subsubsec:T=0} for details). 
The general picture in all these works is that the universal behavior appears in a limit where ``superconductivity is weak'', meaning that $T_c$ is small.

The predicted universality at \emph{positive} temperature is notably less studied. 
It is expected that the ratio $\Xi(T)/T_c$ is given by some universal function of the relative temperature $T/T_c$ \cite{leggett2006quantum,Langmann.Triola.2023},
see \Cref{fig:1}.
For three-dimensional superconductors,\footnote{In \cite{Langmann.Triola.2023}, only the three-dimensional case is considered explicitly. However, their arguments seem to be easily extendable to handle also the cases of one and two-dimensional superconductors.} this has recently been shown in \cite{Langmann.Triola.2023} (building on ideas of \cite{langmann.et.al.2019}) 
for temperatures uniformly in an interval $[0,(1-\eps)T_c]$ for any $\eps > 0$
in an appropriate limit where $T_c$ is small. 
The perhaps most interesting regime of temperatures, however, are those close to the critical temperature, due to the phase transition occurring there. 
For such temperatures one expects\footnote{Historically, the first article suggesting the square root behavior near $T_c$ is by Buckingham \cite{Buckingham}. In \cite[Eq.~(3.31)]{bcs.original}, BCS verified this in their original model, however, with the numerical constant given by $3.2$ instead of $C_{\rm univ} \approx 3.06$. } the behavior \cite[Eq.~(3.54)]{tinkham2004introduction}
\begin{equation}\label{eqn.univ.sqrt.intro}
\Xi(T)/T_c \approx C_{\textnormal{univ}} \sqrt{1 - T/T_c},
\qquad 
C_{\textnormal{univ}} = \sqrt{\frac{8\pi^2}{7\zeta(3)}} \approx 3.06.
\end{equation}
Notably, the critical exponent $1/2$ (i.e.~the order parameter $\Delta(\sqrt\mu)\approx \Xi$ vanishing as a square root) agrees with the prediction from the phenomenological Landau theory \cite{Landau:1937obd} for second order phase transitions (not to be confused with \emph{Ginzburg}-Landau theory of superconductivity \cite{Ginzburg:1950sr, Gorkov.1959, de2018superconductivity}) in mean-field systems.

\begin{figure}[htb]
	\centering
    \includegraphics[width=0.8\textwidth]{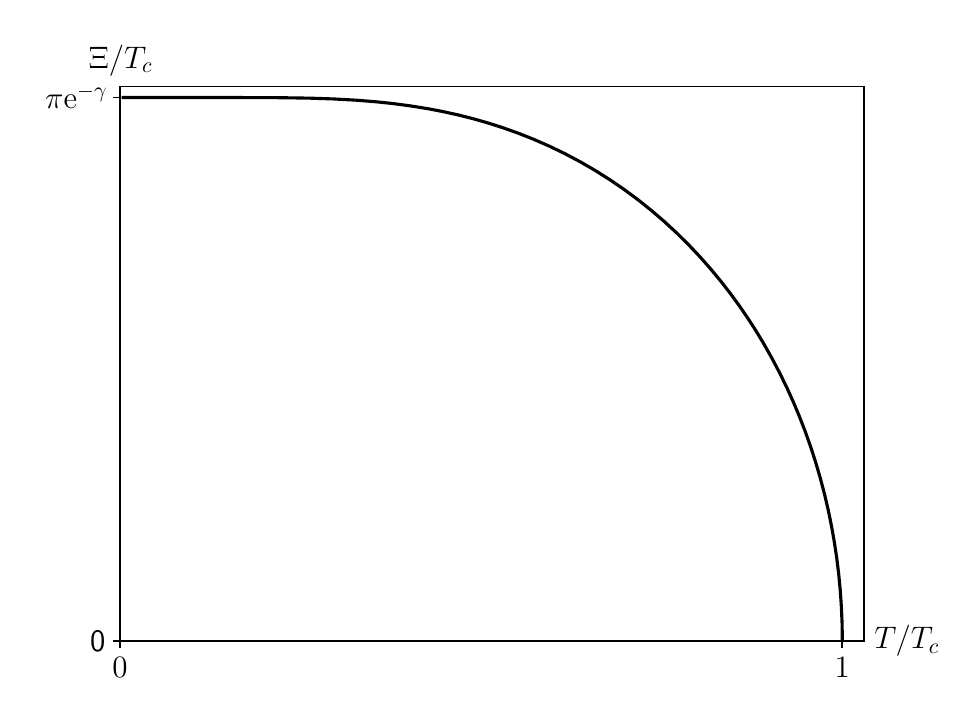}
	\caption[]{The ratio of the BCS energy gap and the critical temperature, $\Xi/T_c$, is well approximated by a \emph{universal function} of the relative temperature $T/T_c$, which is given by $\fBCS(\sqrt{1-T/T_c})$ with $\fBCS$ defined in \eqref{eq:fBCS} below. At $T=0$, it approaches the well-known constant $\pi \mathrm{e}^{-\gamma} \approx 1.76$, with $\gamma \approx 0.57$ being the Euler-Mascheroni constant.}
	\label{fig:1}
\end{figure}

In this paper we extend the previously shown universality in three important directions: Firstly, we consider all spatial dimensions $d \in \{1,2,3\}$. Secondly, we treat the full range of temperatures $0\leq T \leq T_c$. Thirdly, we extend the result to the case of non-zero angular momentum in two dimensions, in particular proving the formula in \eqref{eqn.univ.sqrt.intro}.
Interestingly the case of non-zero angular momentum in two dimensions has the exact same universal behavior as $s$-wave superconductors in any dimensions: Independently of the angular momentum we find the same universal function describing the ratio $\Xi(T)/T_c$.
This is substantially different from the three-dimensional case, where one still expects some sort of universal behavior to occur, 
only the universal function strongly depends on the angular momentum, see, e.g., \cite{Poole.Farach.ea.2007} and Remark \ref{rmk:3dangmom} below.

One of the central ideas in the analysis of temperatures close to the critical temperature is the use of Ginzburg-Landau (GL) theory.
In the physics literature it is well-known that for temperatures close to the critical 
BCS theory is well-approximated by GL theory \cite{Gorkov.1959}.
This correspondence has been studied, and put on rigorous grounds, quite recently in the mathematical physics literature
\cite{Frank.Hainzl.ea.2012,Frank.Lemm.2016, deuchert2023microscopic, deuchert2023microscopic2}. See Section \ref{subsubsec:GLtheory} for more details.

\subsection{Mathematical formulation of BCS theory}
We consider a gas of fermions in $\R^d$ for $d = 1,2,3$ at temperature $T > 0$ and chemical potential $\mu > 0$. The interaction is described by a two-body, real-valued and reflection-symmetric potential $V \in L^1(\R^d)$, for which we assume the following. 
\begin{assumption} \label{ass:basic}
	We have that $V \in L^{p_V}(\R^d)$ for $p_V =1$ if $d = 1$, $p_V \in (1, \infty)$ if $d = 2$, or $p_V = 3/2$ if $d = 3$. 
\end{assumption}
A BCS state $\Gamma$ is given by a pair of functions $(\gamma, \alpha)$ and can be conveniently represented as a $2 \times 2$ matrix valued Fourier multiplier on $L^2(\R^d) \oplus L^2(\R^d)$ of the form
\begin{equation} \label{eq:Gamma}
	\hat{\Gamma}(p) = \begin{pmatrix}
		\hat{\gamma}(p) & \hat{\alpha}(p)\\ \overline{\hat{\alpha}(p)} & 1- \hat{\gamma}(p)
	\end{pmatrix}
\end{equation}
for all $p \in \R^d$. Here, $\hat{\gamma}(p)$ denotes the Fourier transform of the one particle density matrix and $\hat{\alpha}(p)$ is the Fourier transform of the Cooper pair wave function. We require reflection symmetry of $\hat{\alpha}$, i.e.~$\hat{\alpha}(-p) = \hat{\alpha}(p)$, as well as $0 \le \hat{\Gamma}(p) \le 1$ as a matrix. 
Recall the definition of the BCS free energy functional 
\cite{hainzl.seiringer.review.08,hainzl.hamza.seiringer.solovej,hainzl.seiringer.16,Hainzl.Seiringer.2008,hainzl.seiringer.scat.length,Henheik.Lauritsen.2022,Henheik.2022,lauritsen.energy.gap.2021}, which is given by
\begin{equation} \label{eq:functional}
	\mathcal{F}_T[\Gamma] := \int_{\R^d}(p^2 - \mu) \hat{\gamma}(p) \D p - T S[\Gamma] + \int_{\R^d} V(x) |\alpha(x)|^2 \D x\,,
\end{equation}
where the entropy per unit volume is defined as 
\begin{equation*}
	S[\Gamma] = - \int_{\R^d} \mathrm{Tr}_{\C^2} \left[\hat{\Gamma}(p) \log \hat{\Gamma}(p)\right] \D p\,.
\end{equation*}
The variational problem associated with the BCS functional is studied on 
\begin{equation*}
	\mathcal{D} := \left\{   \Gamma \text{ as in } \eqref{eq:Gamma} : 0 \le \hat{\Gamma} \le 1\,, \ \hat{\gamma} \in L^1(\R^d, (1 + p^2) \D p)\,, \ \alpha \in H^1_{\rm sym}(\R^d)  \right\}\,.
\end{equation*}
The following proposition provides the foundation for studying this problem. 
\begin{prop}[{\cite{hainzl.hamza.seiringer.solovej}, see also \cite{hainzl.seiringer.16}}] \label{prop:exofmin}
	Under Assumption \ref{ass:basic} on $V$, the BCS free energy is bounded below on $\mathcal{D}$ and attains its minimum. 
\end{prop}
However, in general, the minimizer is not necessarily unique. This potential non-uniqueness shall not bother us at this stage but will be of importance later on 
(see \Cref{subsubsec:TcXiintro,sec.2d.ell.ne.0}). The Euler--Lagrange equation for $\alpha$ associated with the minimization problem is the celebrated \emph{BCS gap equation}
\begin{equation} \label{eq:gapeq}
	\Delta(p) = - \frac{1}{(2\pi)^{d/2}} \int_{\R^d} \hat{V}(p-q) \frac{\Delta(q)}{K_T^\Delta(q)} \D q\,,
\end{equation}
satisfied by $\Delta(p) = - 2\,  (2 \pi)^{-d/2} (\hat{V}\star \hat{\alpha})(p)$, 
where $\alpha$ is the off--diagonal entry of a minimizing $\Gamma \in \mathcal{D}$ of \eqref{eq:functional}, see \cite{hainzl.hamza.seiringer.solovej,hainzl.seiringer.16}. 
Here, $\hat{V}(p) = (2 \pi)^{-d/2} \int_{\R^d} V(x) \E^{-\I px} \D x$ denotes the Fourier transform of $V$, and we have introduced the notation 
\begin{equation*}
	K_T^\Delta(p) = \frac{E^\Delta(p)}{\tanh\left(\frac{E_\Delta(p)}{2T}\right)} \qquad \text{with} \qquad 
	E_\Delta(p) = \sqrt{(p^2 - \mu)^2 + |\Delta(p)|^2}\,.
\end{equation*}
The gap equation can equivalently be written as
\begin{equation} \label{eq:gapeq2}
	(K_T^\Delta + V) \alpha = 0\,,
\end{equation}
where $K_T^\Delta(p)$ is understood as a multiplication operator in momentum space and $V(x)$ is understood as a multiplication operator in position space. The Euler--Lagrange equation for $\gamma$ (see \cite{hainzl.hamza.seiringer.solovej, hainzl.seiringer.16}) is simply given by 
\begin{equation} \label{eq:ELeqgamma}
\hat{\gamma}(p) = \frac{1}{2} - \frac{p^2 - \mu}{2 K_T^\Delta(p)}\,.
\end{equation}

\subsubsection{Critical temperature and energy gap} \label{subsubsec:TcXiintro}The system described by $\mathcal{F}_T$ is \emph{superconducting} if and only if any minimizer $\Gamma$ of $\mathcal{F}_T$ has off--diagonal entry $\alpha \not\equiv 0$ (or, equivalently, \eqref{eq:gapeq} has a solution $\Delta \not\equiv 0$). The question, whether a system is superconducting or not can be reduced to a linear criterion involving the pseudo--differential operator with symbol
\begin{equation*}
	K_T(p) \equiv K_T^0(p) = \frac{p^2 - \mu}{\tanh\left(\frac{p^2 - \mu}{2T}\right)}\,.
\end{equation*}
In fact, as shown in \cite{hainzl.hamza.seiringer.solovej}, 
the system is superconducting if and only if the operator $K_T + V$ has at least one negative eigenvalue. Moreover, there exists a unique \emph{critical temperature} $T_c \ge 0$ being defined as 
\begin{equation} \label{eq:Tc}
	\boxed{T_c := \inf\{ T > 0 : K_T + V \ge 0 \}}\,,
\end{equation}
for which $K_{T_c} + V \ge 0$ and $\inf \mathrm{spec} (K_T + V) < 0$ for all $T < T_c$. By Assumption \ref{ass:basic} and the asymptotic behavior $K_{T_c}(p) \sim p^2$ for $|p| \to \infty$ the critical temperature is well-defined by Sobolev's inequality \cite[Thm.~8.3]{analysis}. Note that, for $T \ge T_c$ the BCS functional \eqref{eq:functional} is uniquely minimized by the pair $\Gamma_{\rm FD} \equiv (\gamma_{\rm FD}, 0)$, where 
\begin{equation} \label{eq:FDdistr}
	\hat{\gamma}_{\rm FD}(p) = \frac{1}{1 + \E^{\frac{1}{T} (p^2 - \mu)}}
\end{equation}
is the usual Fermi-Dirac distribution. In contrast, for temperatures $0 \le T < T_c$ strictly below the critical temperature, 
it is a priori not clear whether or not the minimizer of \eqref{eq:functional} is unique.

In this paper we deal with two different cases.
In the case of $s$-wave superconductivity we will assume properties of $V$ such that the minimizer is unique 
and in the case of $2$-dimensional non-zero angular momentum we will assume properties of $V$ 
such that there are at most $2$ minimizers, see \Cref{sec.2d.ell.ne.0}.

For the $s$-wave case we assume the following.

\begin{assumption} \label{ass:hatV neg}
	Let the (real valued) interaction potential $V \in L^1(\R^d)$ be \emph{radially symmetric} and assume that $V$ is of \emph{negative type}, i.e.~
	$\hat{V} \le 0$ and $\hat{V}(0) < 0$. 
\end{assumption}
As shown in \cite{Hainzl.Seiringer.2008}, Assumption \ref{ass:hatV neg} implies that, in particular, the critical temperature is non-zero, i.e.~$T_c > 0$.\footnote{To be precise, the arguments in \cite{Hainzl.Seiringer.2008} cover only the case $d=3$, but, as already noted in \cite{Frank.Hainzl.ea.2012}, they are immediately transferable to the cases $d = 1, 2$.}
Moreover, as already indicated above, it ensures that the minimizer of \eqref{eq:functional} is unique. 
While this fact is already known at zero temperature \cite[Lemma~2]{Hainzl.Seiringer.2008}, 
we are not aware of any place in the literature where the extension to positive temperature is given.
As we will need this extension, we formulate it in the following proposition
and give a proof in \Cref{app:Deltaunique}.
\begin{prop}[Uniqueness of minimizers for potentials of negative type] \label{prop.KTDelta.geq.0} 
	Let $V$ satisfy Assumptions~\ref{ass:basic} and \ref{ass:hatV neg}, and consider the BCS functional \eqref{eq:functional}.  Then we have the following: 
	\begin{itemize}
		\item[(i)] For $0 \le T < T_c$, let $ \Gamma \equiv (\gamma, \alpha)$ be a minimizer of the BCS functional \eqref{eq:functional} (which exists by means of Proposition \ref{prop:exofmin}). Then the operator $K_T^\Delta + V$ from \eqref{eq:gapeq2} is non--negative and $\alpha$ is its unique ground state with eigenvalue zero $0$. 
		\item[(ii)] The minimizer $\Gamma  =: \Gamma_* \equiv (\gamma_*, \alpha_*)$ of \eqref{eq:functional} is unique up to a phase of $\alpha_*$ and can be chosen to have strictly positive Fourier transform $\hat{\alpha}_*$. Moreover, both $\gamma_*$ and $\alpha_*$ are radial functions. 
	\end{itemize}
\end{prop}

In particular, under Assumption \ref{ass:hatV neg}, we have that the \emph{energy gap}
\begin{equation} \label{eq:energygap}
\boxed{	\Xi(T) := \inf_{p \in \R^d} \sqrt{(p^2 - \mu)^2 + |\Delta(p)|^2}}\,,
\end{equation}
for $\Delta$ being the unique non-zero solution of \eqref{eq:gapeq} and temperatures $0 \le T< T_c$, is well-defined.

In case there is more than one solution $\Delta$ of the BCS gap equation \eqref{eq:gapeq} (i.e. more than one minimizer of the BCS functional)
we may for each such $\Delta$ define the energy gap $\Xi$ as in \Cref{eq:energygap}. 
In the case of two dimensions with (definite) non-zero angular momentum 
we shall prove that there exist exactly two such functions, $\Delta_{\pm}$.
They however satisfy $|\Delta_+| = |\Delta_-|$ and so the energy gap $\Xi$ is also here uniquely defined. 
For the details see \Cref{sec.2d.ell.ne.0}.

\begin{remark}\label{remark.gap=order-parameter}
The energy gap is essentially the same as the order parameter $\abs{\Delta(\sqrt\mu)}$ as we show in 
\Cref{eqn.compare.Delta.Xi,eq:Xi=Delta} below. 
In particular, one may replace $\Xi$ with $\abs{\Delta(\sqrt\mu)}$ in our main results, \Cref{prop:BCSoriguniv} and \Cref{thm:main,thm:main2}.
\end{remark}

\subsubsection{Weak coupling}
We consider here the weak--coupling limit where the interaction is of the form $\lambda V$ for a $\lambda > 0$ and we consider the limit $\lambda \to 0$.
In the weak--coupling limit an important role is played by the (rescaled) operator 
$\mathcal{V}_{\mu} : L^2(\mathbb{S}^{d-1}) \to L^2(\mathbb{S}^{d-1})$ \cite{Hainzl.Seiringer.2008, Henheik.Lauritsen.ea.2023,cuenin.merz.2021,hainzl.seiringer.2010}.
This operator, which is defined as
\begin{equation} \label{eq:Vmu}
  \left(\mathcal{V}_{\mu} u\right)(p) = \frac{1}{(2\pi)^{d/2}} \int_{\Sph^{d-1}} \hat{V}(\sqrt{\mu} (p-q)) u(q)\, \D\omega(q) \,,
\end{equation}
where $\D \omega$ denotes the uniform (Lebesgue) measure on the unit sphere $\Sph^{d-1}$, measures the strength of the interaction potential $\hat{V}$ on the Fermi surface. 
The pointwise evaluation of $\hat{V}$ (and in particular on a $\mathrm{codim-}1$ submanifold) 
is well-defined since we assume that $V \in L^1(\R^d)$.

The lowest eigenvalue $e_{\mu} = \mathrm{inf\, spec}\, \mathcal{V}_{\mu}$ is of particular importance. 
Note, that $\mathcal{V}_{\mu}$ is a trace-class operator (see the argument above \cite[Equation (3.2)]{frank.hainzl.naboko.seiringer}) with 
\begin{equation*}
  \mathrm{tr}(\mathcal{V}_{\mu}) = \frac{|\Sph^{d-1}|}{(2\pi)^d} \int_{\R^d} V(x) \D x = \frac{|\Sph^{d-1}|}{(2\pi)^{d/2}} \, \hat{V}(0)\,.
\end{equation*} 
For radial potentials one sees that the eigenfunctions of $\mcV_\mu$ are the spherical harmonics.

For potentials of negative type we have $\hat V(0) < 0$ and so $e_\mu < 0$. 
This corresponds to an attractive interaction between (some) electrons on the Fermi sphere.
Further, one easily sees that the constant function $u(p) = (|\Sph^{d-1}|)^{-1/2}$ is an eigenfunction of $\mathcal{V}_{\mu}$, which, since $\hat{V} \le 0$ by Assumption~\ref{ass:hatV neg}, is in fact the ground state by the Perron--Frobenius theorem, i.e.
\begin{equation} \label{eq:emu}
  e_{\mu} = \frac{1}{(2\pi)^{d/2}} \int_{\Sph^{d-1}} \hat{V}(\sqrt{\mu} - q \sqrt{\mu})  \, \D \omega(q)\,.
\end{equation}

In two dimensions the spherical harmonics take the form $u_{\pm\ell}(p) = (2\pi)^{-1/2} e^{\pm i \ell \varphi}$ 
with $\varphi $ denoting the angle of $p\in \R^2$ in polar coordinates. 
In this case the ground state has some definite angular momentum $\ell_0$. 
If $\ell_0\ne 0$ then the ground state is at least twice degenerate, since then both $u_{\pm\ell_0}$ 
are eigenfunctions of this lowest eigenvalue.

\subsection{Previous mathematical results}
So far, all mathematical results on solutions of the BCS gap equation \eqref{eq:gapeq} focused either on zero temperature, $T = 0$, or the regime close to the critical one, $T \approx T_c$, where the transition from superconducting to normal behavior is described by Ginzburg-Landau theory. 
\subsubsection{BCS theory in limiting regimes: Universality at \texorpdfstring{$T=0$}{T=0}} \label{subsubsec:T=0}
At zero temperature it is expected, that the ratio of the energy gap and the critical temperature is given by a universal constant, 
\begin{equation} \label{eq:T=0univ}
	\frac{\Xi(T=0)}{T_c}\approx \pi e^{-\gamma}\,,
\end{equation}
with $\gamma\approx 0.577$ the Euler--Mascheroni constant
in a limiting regime where ``superconductivity is weak'', meaning that $T_c$ is small.

In the literature three such limits have been studied: Historically, the first regime, which has been considered is the 
\emph{weak coupling limit} in three spatial dimensions \cite{Hainzl.Seiringer.2008,frank.hainzl.naboko.seiringer}, which we recently extended to one and two dimensions in \cite{Henheik.Lauritsen.ea.2023}. The critical temperature in the \emph{low density limit} in three dimensions was studied in \cite{hainzl.seiringer.scat.length} and later complemented by a study of the energy gap by one of us in \cite{lauritsen.energy.gap.2021}, thus, in combination, yielding the above-mentioned universal behavior. 
Finally, we considered the \emph{high density limit}, again in three dimension, in \cite{Henheik.2022,Henheik.Lauritsen.2022} and proved \eqref{eq:T=0univ} in this regime.

\subsubsection{Superconductors close to \texorpdfstring{$T_c$}{Tc}: Ginzburg-Landau theory} \label{subsubsec:GLtheory}
For temperatures close to the critical BCS theory is well-approximated by Ginzburg-Landau (GL) theory. In contrast to the microscopic BCS model, GL theory is a phenomenological model, which describes the superconductor on a macroscopic scale. Moreover, as suggested by \Cref{eqn.univ.sqrt.intro} a natural parameter measuring ``closeness to $T_c$''
is the parameter $h = \sqrt{1 - T/T_c}$. 
A rigorous analysis of various aspects of BCS theory in the limit $h\to 0$ was then studied in 
\cite{Frank.Hainzl.ea.2012,Frank.Lemm.2016,Frank.Hainzl.ea.2016}, very recently also allowing for general external fields \cite{deuchert2023microscopic, deuchert2023microscopic2}. 
Of particular interest to us is the fact that any minimizer of the BCS functional $(\gamma, \alpha)$
has $\alpha \approx h \psi \mathfrak{a}_0$
with $\mathfrak{a}_0 \in \ker (K_{T_c} + \lambda V)$ fixed and $\psi \in \C$ a minimizer of the corresponding GL functional, see \cite[Theorem 2.10]{Frank.Lemm.2016}.

\subsection{Outline of the paper}
The rest of this paper is structured as follows. In Section \ref{sec:results} we present our main result, starting with the prototypical universality in the original BCS model (Section \ref{sec.bcs.original}). Afterwards, in Sections~\ref{subsec:swave} and~\ref{sec.2d.ell.ne.0} we describe our results on universality for $s$-wave superconductors in arbitrary dimension $d \in \{1,2,3\}$, and for two-dimensional superconductors having pure angular momentum, respectively. The proofs of these results are given in Section \ref{sec:proofs}, while several additional proofs are deferred to Appendix \ref{app:auxandadd}.

\section{Main Result} \label{sec:results}
We next describe the main results of the paper. 
We first consider the example of an interaction as considered by BCS \cite{bcs.original}.
The reason for doing this is twofold: 
\begin{enumerate}
\item 
It highlights, how the universal function $\Xi(h)/T_c \approx \fBCS(h)$ appears.
\item 
A central idea in the proof 
of removing the $\log$-divergence is already present in the BCS gap equation \eqref{eq:gapeq}.
\end{enumerate}
For $T=\Delta=0$ we have $K_{T=0}^{\Delta=0}(p) = |p^2-\mu|$. 
This gives rise to a logarithmic-divergence in \Cref{eq:gapeq}. 
Understanding how to treat this $\log$-divergence was one of the key insights of Langmann and Triola \cite{Langmann.Triola.2023}.

\subsection{Energy gap in the original BCS approximation \texorpdfstring{\cite{bcs.original}}{}}
\label{sec.bcs.original}
In their seminal work \cite{bcs.original}, Bardeen--Cooper--Schrieffer modeled the interaction by a so called \emph{separable potential} $V(x,y)$ (i.e.~factorizing and depending not only on the difference variable $x-y$), whose Fourier transform $\hat{V}(p,q) $ 
is a product of two \emph{radial single variable} functions, that are 
compactly supported in the shell 
\begin{equation} \label{eq:bcsshell}
\mathcal{S}_\mu(T_D) := \{ p \in \R^d : |p^2 - \mu| \le T_D\}
\end{equation}
around the Fermi surface $\{ p \in \R^d : p^2 = \mu\}$, the only (material dependent) parameter being the so-called \emph{Debye temperature} $0 <T_D< \mu$. Switching from momentum $p$ to energy $\epsilon = p^2 - \mu$, the just mentioned single variable functions are chosen in such a way, that\footnote{Assuming that $\hat{V}$ is constant throughout the energy shell \eqref{eq:bcsshell} (as done in \cite{bcs.original}), the BCS coupling parameter emerges as $\lambda_{\rm BCS} = -\hat{V}(0,0) N(0)$.}
\begin{equation*}
\hat{V}(\epsilon, \epsilon') N(\epsilon') = - \lambda_{\rm BCS} \, \theta(1- |\epsilon/T_D|) \theta(1- |\epsilon'/T_D|)\,, \qquad \lambda_{\rm BCS} > 0\,,
\end{equation*}
where the electronic density of states (DOS) is denoted by $N(\epsilon) \sim (\epsilon + \mu)^{(d-2)/2}$ and $\theta$ is the Heaviside function.
($\theta(t) = 1$ for $t>0$ and $\theta(t)=0$ otherwise.)

In this case, the (unique non-negative) solution to the BCS gap equation \eqref{eq:gapeq} is given by 
\begin{equation} \label{eq:BCSgap}
\Delta(\epsilon) = \Delta \cdot \theta(1- |\epsilon/T_D|)
\end{equation}
for some temperature dependent constant $\Delta \ge 0$,
which is determined by the scalar gap equation (cf.~\cite[Eq.~(3.27)]{bcs.original})
\begin{equation} \label{eq:scalargapeq}
\frac{1}{\lambda_{\rm BCS}} = \int_{0}^{T_D} \frac{\tanh\big(\frac{\sqrt{\epsilon^2 + \Delta^2}}{2T}\big)}{\sqrt{\epsilon^2 + \Delta^2}} \mathrm{d} \epsilon 
\end{equation}
for any temperature $0 \le T < T_c$. In turn, the critical temperature $T_c>0$ is determined by \eqref{eq:scalargapeq} with $\Delta = 0$, i.e.
\begin{equation} \label{eq:scalargapeqTc}
\frac{1}{\lambda_{\rm BCS}} = \int_{0}^{T_D} \frac{\tanh\big(\frac{\epsilon}{2T_c}\big)}{\epsilon} \mathrm{d} \epsilon \,.
\end{equation}
In case of a small BCS coupling parameter, $\lambda_{\rm BCS} \ll 1$,\footnote{This can happen for various reasons. One example is that $V$ itself is scaled by a coupling parameter $\lambda > 0$, i.e.~$V \to \lambda V$, and one considers the limit $\lambda \to 0$, as done in Sections \ref{subsec:swave}--\ref{sec.2d.ell.ne.0}.} it holds that $T_c$ is exponentially small in $\lambda_{\rm BCS}$, i.e.~$T_c \sim \ee^{-1/\lambda_{\rm BCS}}$ (see \cite[Eq.~(3.29)]{bcs.original}). Moreover, it is easily checked that $\Delta$ as a function temperature is monotonically decreasing in the interval $[0,T_c]$ and satisfies $\Delta(T=0) \sim \ee^{-1/\lambda_{\rm BCS}}$, similarly to the critical temperature.

Next, changing variables as $x:= \epsilon/T_c$ and setting $\delta := \Delta/T_c$ as well as\footnote{As mentioned above, the parameter $h$ is commonly used (see, e.g., \cite{Frank.Hainzl.ea.2012, Frank.Lemm.2016}) in the context of Ginzburg-Landau theory, where it served as a `semiclassical' small parameter in the derivation this theory.}
\begin{equation} \label{eq:defh}
	h := \sqrt{1- \frac{T}{T_c}} \quad \text{for} \quad 0 \le T \le T_c\,,
\end{equation}
we can subtract \eqref{eq:scalargapeq} and \eqref{eq:scalargapeqTc} to find
\begin{equation} \label{eq:diffformula}
\int_{0}^{T_D/T_c} \left\{ \frac{\tanh\big(\frac{\sqrt{x^2 + \delta^2}}{2(1 - h^2)}\big)}{\sqrt{x^2 + \delta^2}}  - \frac{\tanh\big(\frac{x}{2}\big)}{x} \right\} \mathrm{d}x = 0 \,. 
\end{equation}
Note that this difference formula \eqref{eq:diffformula} removes the divergences of \eqref{eq:scalargapeq}--\eqref{eq:scalargapeqTc} as $\lambda_{\rm BCS} \to 0$. 

The proof of the following proposition is given in Section \ref{subsec:proofBCSorig}.
(In the statement of \Cref{prop:BCSoriguniv}, one may replace $\Xi$ by the order parameter $\Delta(\sqrt\mu)$, see \Cref{remark.gap=order-parameter} above.)
\begin{prop}[Energy gap in the original BCS model \cite{bcs.original}]
\label{prop:BCSoriguniv}
Let $\mu> 0$, fix a Debye temperature $0 < T_D < \mu$ and let $\lambda_{\rm BCS}> 0$ be the BCS coupling parameter as above. Let the critical temperature $T_c$ and the gap function $\Delta(p)$ be defined via \eqref{eq:BCSgap}--\eqref{eq:scalargapeqTc}. 

Then the energy gap $\Xi$ (defined in \eqref{eq:energygap}) as a function of $h = \sqrt{1 - T/T_c}$ for $0 \le T \le  T_c$ (recall \eqref{eq:defh}) is given by 
\begin{equation} \label{eq:XiunivBCSorig}
\Xi(h) = T_c \, \fBCS(h) \, \big(1 + O(\ee^{-1/\lambda_{\rm BCS}})\big)
\end{equation}
uniformly in $h \in [0,1]$, where the function $\fBCS :[0,1] \to [0,\infty)$ is implicitly defined via
\begin{equation} \label{eq:fBCS}
	\int_\R 
	\left\{\frac{\tanh \frac{\sqrt{s^2 + \mathsf{f}_{\textnormal{BCS}}(h)^2}}{2(1-h^2)}}{\sqrt{s^2 + \mathsf{f}_{\textnormal{BCS}}(h)^2}}  
	-  \frac{\tanh \frac{s}{2}}{s} \right\} \ud s =0 
\end{equation}
and plotted in Figure  \ref{fig:2}. 
\end{prop}
This means that, independent of the material dependent Debye temperature $T_D > 0$ and the chemical potential $\mu >0$, the energy gap $\Xi$ within the original BCS approximation \cite{bcs.original}, follows a universal curve, described by \eqref{eq:XiunivBCSorig}, in the limit of weak BCS coupling. A similar formula for $\fBCS$ like \eqref{eq:fBCS} (but as a function of $x:=1-h^2$) also appeared in the monograph of Leggett \cite[Eq.~(5.5.21)]{leggett2006quantum}. We now list a few basic properties of $\fBCS$, whose proofs we omit, as they can be obtained by means of the implicit function theorem and further elementary tools (see also \cite[Lemma~1]{Langmann.Triola.2023} as well as Lemmas \ref{lem:delta(h)apriori} and \ref{lem:Cuniv} below). Almost all of these properties become apparent from Figure \ref{fig:2}. 
\begin{lem}[Properties of $\fBCS$] \label{lem:fBCS}
There exists a unique implicitly defined solution function $\fBCS :[0,1] \to [0,\infty)$ of \eqref{eq:fBCS}. Moreover, $\fBCS$ has the following properties:
\begin{itemize}
\item[(i)] It is strictly monotonically increasing in $[0,1]$.
\item[(ii)] It is $C^1$ in $(0,1)$ and has continuous one-sided derivatives at the boundaries $0$ and $1$.
\item[(iii)] It has the boundary values $\fBCS(0)=0$, $\fBCS'(0) = C_{\rm univ}$ and $\fBCS(1) = \pi \ee^{-\gamma} \approx 1.76$, $\fBCS'(1) = 0$.  Here, $\gamma \approx 0.57$ is the Euler-Mascheroni constant and 
\begin{equation} \label{eq:Cuniv}
	C_{\rm univ} := \sqrt{\frac{8 \pi^2}{7 \zeta(3)}} \approx 3.06 \,,
\end{equation}
where $\zeta(s)$ denotes Riemann's $\zeta$-function.
\end{itemize}
\end{lem}

\begin{figure}[htb]
	\centering
  \includegraphics[width=0.8\textwidth]{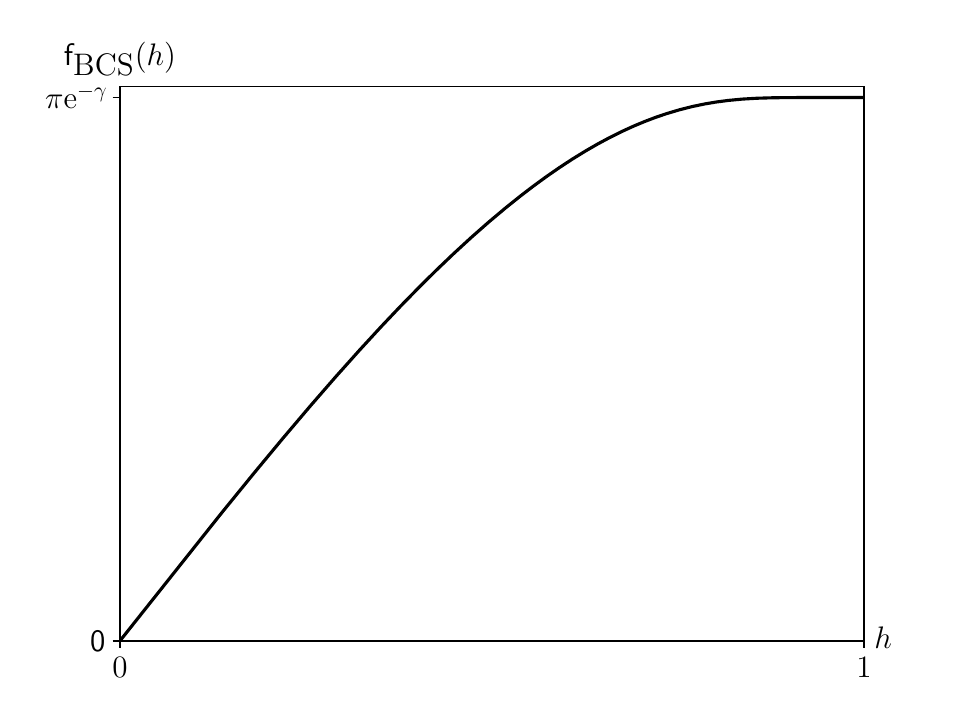}
	\caption[]{Sketch of the function $\fBCS$ obtained via the implicit relation \eqref{eq:fBCS}.}
	\label{fig:2}
\end{figure}

\begin{remark}[Contact interactions] \label{rmk:contact}Our proof of Proposition \ref{prop:BCSoriguniv} can easily be generalized to all BCS models, in which the energy gap is constant (at least near the Fermi surface). 
	\begin{itemize}
\item[(a)] In case of a delta potential, $V(x) = - \delta(x)$ in one spatial dimension, $d=1$, the gap function solving \eqref{eq:gapeq} is given by a constant (simply because here $\hat{V}$ is constant). This setting can similarly be analyzed (in a weak coupling limit, i.e.~replacing $V \to \lambda V$ and taking $\lambda \to 0$) as done in Proposition \ref{prop:BCSoriguniv} for the original BCS model \cite{bcs.original}.
\item[(b)] Also for contact interactions in three spatial dimensions, $d=3$, the situation is similar. 	This setting is studied in \cite{Braeunlich.Hainzl.ea.2014,braunlich.hainzl.seiringer}, where it is shown that for a suitable sequence of potentials $V_\ell$ 
converging to a point interaction with scattering length $a < 0$, the gap function $\Delta_\ell$ converges (uniformly on compact sets, see \cite[Eq.~(14)]{Braeunlich.Hainzl.ea.2014}) 
to a constant $\Delta$ solving the gap equation 
\[
- \frac{1}{4\pi a} 
= \frac{1}{(2\pi)^{3}} \int_{\R^3} 
\left(\frac{1}{K_T^\Delta(p)}
- \frac{1}{p^2}\right) \ud p\,. 
\]
Replacing the limit of weak coupling by a small scattering range limit, $a \to 0$, one can obtain a result similar to Proposition \ref{prop:BCSoriguniv}.  
	\end{itemize}
\end{remark}

\subsection{Universal behavior of the \texorpdfstring{$s$}{s}-wave BCS energy gap} \label{subsec:swave}
After having discussed the prototypical universality in the seminal BCS paper \cite{bcs.original}, we can now formulate our main result on general $s$-wave superconductors with local interactions. 
The proof of Theorem \ref{thm:main} is presented in Sections \ref{subsec:partaproof}--\ref{subsec:GLproof}, while the main ideas are briefly described in Remark \ref{rmk:proof} below. 
(We remark that, in the statement of \Cref{thm:main}, one may replace $\Xi$ by the order parameter $\Delta(\sqrt\mu)$, see \Cref{remark.gap=order-parameter} above.)
\begin{thm}[BCS energy gap for $s$-wave superconductors]  \label{thm:main}
  Let $d \in \{1,2,3\}$, $\mu > 0$ and let $V$ satisfy \Cref{ass:basic,ass:hatV neg}.
  Define the critical temperature $T_c$ and energy gap $\Xi$ as in \eqref{eq:Tc} and \eqref{eq:energygap} with interaction $\lambda V$ for a $\lambda > 0$
  and $\Delta$ being the unique non-zero solution the BCS gap equation \eqref{eq:gapeq} with interaction $\lambda V$.
	
	Then, with $\fBCS(h), h = \sqrt{1 - T/T_c}$ being the function defined via \eqref{eq:fBCS}, we have the following: 
	\begin{itemize}
\item[(a)] Assuming additionally that $|\cdot|^2 V \in L^1(\R^d)$, it holds that
\begin{equation} \label{eq:maina}
\Xi(h) = T_c \, \fBCS(h) \, \big(1 + O(h^{-1} \ee^{-c/\lambda})\big)
\end{equation}
for some constant $c>0$ independent of $\lambda$ and $h$. 
\item[(b)] Assuming additionally that $(1 + |\cdot|) V \in L^2(\R^d)$, it holds that
\begin{equation} \label{eq:mainb}
	\Xi(h) = T_c \, \fBCS(h) \, \big(1 + O(h \ee^{c'/\lambda}) + o_{\lambda \to 0}(1)\big)
\end{equation}
for some constant $c'>0$ independent of $\lambda$ and $h$ and where $o_{\lambda \to 0}(1)$ vanishes as $\lambda \to 0$ uniformly in $h$.
	\end{itemize}
\end{thm}

For the special case $h=1$, i.e.~$T=0$, \eqref{eq:maina} reproduces the results from \cite{Hainzl.Seiringer.2008} (for $d=3$) and \cite{Henheik.Lauritsen.ea.2023} (for $d=1,2$), which state the universality 
\begin{equation*}
\lim\limits_{\lambda \to 0} \frac{\Xi(T=0)}{T_c} = \frac{\pi}{\ee^{\gamma}}
\end{equation*}
at $T=0$. Moreover, by \eqref{eq:maina} again, we find that, uniformly in temperatures bounded away from $T_c$, i.e.~$h \in [\epsi, 1]$ for some fixed $\epsi >0$,
\begin{equation*}
\lim\limits_{\lambda \to 0} \frac{\Xi(h)}{T_c} = \fBCS(h)\,,
\end{equation*}
recovering the universality result in \cite{Langmann.Triola.2023} (for $d=3$), with an exponential speed $O(\ee^{-c/\lambda})$ of convergence. In the complementary case, for temperatures very close to the critical temperature, $T \approx T_c$, the question of universality is (i) physically more interesting due to the phase transition from superconducting to normal behavior and (ii) mathematically more delicate than in the previous scenarios. This is because now there are \emph{two} small parameters $\lambda$ and $h$, instead of $\lambda$ only, and the error term in \eqref{eq:maina} might actually be large compared to one. However, now involving both, \eqref{eq:maina} \emph{and} \eqref{eq:mainb}, we find that 
 \begin{equation}\label{eq:limit.orders}
\lim\limits_{\substack{ \lambda,h \to 0 \\ \ee^{-c/\lambda} \ll h}}\frac{\Xi(h)}{T_c\, h } = C_{\rm univ} \quad \text{and} \quad \lim\limits_{\substack{ \lambda,h \to 0 \\  h \ll \ee^{-c'/\lambda}}}\frac{\Xi(h)}{T_c\, h } = C_{\rm univ}
 \end{equation}
with the aid of Lemma \ref{lem:fBCS}~(iii). In particular, the ratio $\Xi(h)/(T_c h)$ converges to the \emph{same} universal constant $C_{\rm univ}$ (recall \eqref{eq:Cuniv}) in both orders of limits, $\lim_{\lambda \to 0} \lim_{h \to 0}$ and $\lim_{h \to 0} \lim_{\lambda \to 0}$.

\begin{rmk}[Joint limit] 
A careful inspection of the proof reveals that the constants $c,c'$ satisfy $c < c'$.
In particular, the proof does \emph{not} allow the two regimes considered in \eqref{eq:limit.orders} to be overlapping and we cannot prove that 
$\lim_{\lambda,h\to0}\frac{\Xi(h)}{T_ch } = C_{\textnormal{univ}}$ in any joint limit. 
We expect this to hold in any joint limit, however, as we saw for the particular example from \cite{bcs.original} in \Cref{prop:BCSoriguniv}
\end{rmk}

\begin{rmk}[{Comparison of assumptions with \cite{Langmann.Triola.2023}}]
Compared to the similar result in \cite{Langmann.Triola.2023} our assumptions hold for a slightly different class of potentials. 
The assumptions of \cite{Langmann.Triola.2023} are essentially on the smoothness of the interaction $V$ 
(formulated via some regularity/decay assumption on the Fourier transform $\hat V$).
Our assumptions on the other hand are on the regularity/decay of $V$. 
In particular, our assumptions cover the examples of \cite[Table I]{Langmann.Triola.2023} 
which are \emph{not} covered by the assumptions of \cite{Langmann.Triola.2023}. 
These are (in three dimensions)
\begin{equation*}
V_{\textnormal{Yukawa}}(x) = \frac{e^{-|x|}}{4\pi |x|},
\qquad 
V_{a\textnormal{Y} + b\textnormal{E}}(x) = \frac{(2a+b|x|)e^{-|x|}}{8\pi |x|},
\qquad 
V_{x\textnormal{-box}}(x) = \frac{3\theta(1-|x|)}{4\pi}\,.
\end{equation*}
\end{rmk}

\begin{rmk}[On the proof]\label{rmk:proof} The main ideas in the proof of Theorem \ref{thm:main} are the following. 
\begin{itemize}
\item[(a)] For part (a), we crucially use that both, $K_{T}^\Delta + \lambda V$ and $K_{T_c} + \lambda V$, have lowest eigenvalue zero. We then consider their corresponding Birman-Schwinger (BS) operators and use that, for $\lambda$ small enough, two naturally associated operators on the Fermi sphere both have the \emph{same} ground state. Evaluating the difference of these two associated operators in this common ground state, we find that a \emph{difference} of two logarithmically divergent integrals, similarly to \eqref{eq:fBCS}, vanishes up to exponentially small errors $O(\mathrm{e}^{-c/\lambda})$. 

The removal of the $\log$-divergence in this way (which -- in a similar fashion -- was the major insight in \cite{Langmann.Triola.2023}) is the key idea to (i) access also non-zero temperatures and (ii) obtain extremely precise error estimates (compared to all the previous results mentioned in Section \ref{subsubsec:T=0}). 
\item[(b)] For part (b), we employ Ginzburg-Landau (GL). The principal realm of GL theory is to describe superconductors and superfluids close to their critical temperature $T_c$. In this regime, when superconductivity is weak, the main idea is that the prime competitor for developing a small off-diagonal component $\hat{\alpha}$ for a BCS minimizer, is the \emph{normal state} $\Gamma_{\rm FD} = (\gamma_{\rm FD}, 0)$, with  $\gamma_{\rm FD}$ given by the usual Fermi-Dirac distribution (recall \eqref{eq:FDdistr}). Moreover, to leading order, the off-diagonal component $\hat{\alpha}$ lies in the kernel (which agrees with the ground state space) of the operator $K_{T_c} + \lambda V$.

The main input, which we use, is that every minimizer $(\gamma, \alpha)$ of the BCS functional has $\alpha \approx h \psi \mathfrak{a}_0$ with $\mathfrak{a}_0 \in \ker (K_{T_c} + \lambda V)$ fixed and $\psi \in \C$ minimizing the corresponding GL functional \cite[Theorem 2.10]{Frank.Lemm.2016}. Taking the convolution of $\hat{\mathfrak{a}}_0$ with $\hat{V}$, we find the universal constant \eqref{eq:Cuniv} appearing in $\Xi/(T_c h) \approx C_{\rm univ}$. 
\end{itemize}
Moreover, the ``additional assumptions" in Theorem \ref{thm:main} are not quite rigid, meaning that they can be weakened in the following sense. 
	\begin{itemize}
		\item[(a)] 
In case that $|\cdot|^{2\alpha} V \in L^1$ for a $0 < \alpha \leq 1$ the error term in \Cref{eq:maina} should instead be $O(h^{-1}e^{-c\alpha/\lambda})$ 
with the constant $c$ then being independent also of $\alpha$.

\item[(b)] In case that $(1 + |\cdot|) V \in L^p(\R^d)$ for $p < 2$, the factor $h$ in the first error term in \eqref{eq:mainb} would not appear raised to the first power but with exponent
\begin{equation*}
(3p-4)/p \quad \text{for} \quad d=1\,, \quad (3p-4)/p -\epsilon\quad \text{for} \quad d=2\,, \quad \text{and} \quad (4p-6)/p \quad \text{for} \quad d=3\,. 
\end{equation*}
	\end{itemize}
\end{rmk}

\begin{rmk}[Other limits] 
Although, in this paper, we considered only the weak coupling limit, we expect the relation $\Xi(h) \approx T_c \fBCS(h)$ to hold also in other limiting regimes in which ``superconductivity is weak", that is, e.g., the low-\footnote{For dimensions $d=1,2$, the same caveats mentioned in \cite[Remark~2.8]{Henheik.Lauritsen.ea.2023} apply.} and high-density limit, that were studied in \cite{hainzl.seiringer.scat.length, lauritsen.energy.gap.2021} and \cite{Henheik.2022, Henheik.Lauritsen.2022}, respectively. This idea is already contained in \cite{Langmann.Triola.2023}, where the authors considered a ``universal" parameter $\lambda$ in \cite[Eq.~(7)]{Langmann.Triola.2023}, which can be small for various physical situations. 
\end{rmk}

\begin{rmk}[Non-universality]
In the proof of part (a) we recover also the formula \cite[Equation (16)]{Langmann.Triola.2023}, 
that
\begin{equation}\label{eqn.nonuniversality}
\frac{\Delta(p)}{\Delta(\sqrt\mu)} = F(p) + O(e^{-c/\lambda})
\end{equation}
for some function $F$ not depending on the temperature and some constant $c>0$. It depends on the interaction $V$ however, 
hence why this is called a ``non-universal'' feature in \cite{Langmann.Triola.2023}. The proof of \eqref{eqn.nonuniversality} is given in Section \ref{subsec:nonuniv}. 
\end{rmk}

\subsection{The case of pure angular momentum for \texorpdfstring{$d=2$}{d=2}}\label{sec.2d.ell.ne.0}
In this section, we generalize Theorem \ref{thm:main} from $s$-wave superconductors to two-dimensional systems which have a definite (or \emph{pure}) angular momentum $\ell_0 \in \N_0$, which can differ from~$0$. 
\begin{assumption}[Pure angular momentum] \label{ass:pureell}
Let $V \in L^1(\R^2)$ be radially symmetric and \emph{attractive on the Fermi sphere}, i.e.~the lowest eigenvalue $e_\mu$ of $\mathcal{V}_\mu$ is strictly negative (recall \eqref{eq:Vmu}--\eqref{eq:emu}). Moreover, suppose that for all $\lambda > 0$ small enough the lowest eigenvalue of $K_{T_c} + \lambda V$ is at most twice degenerate, i.e.~$\dim \ker (K_{T_c} + \lambda V) \in \{1,2\}$. 
\end{assumption}
Since $K_{T_c}$ commutes with the Laplacian, Assumption \ref{ass:pureell} ensures the ground state of $K_{T_c} + \lambda V$ to have definite angular momentum. More precisely, it holds that 
\begin{equation} \label{eq:pureang}
\ker(K_{T_c} + \lambda V) = \mathrm{span} \{\rho\} \otimes \mathcal{S}_{\ell_0} \quad  \text{with} \quad \mathcal{S}_{\ell} := \mathrm{span}\big\{ \mathrm{e}^{\pm \mathrm{i} \ell \varphi} \big\} \subset L^2(\Sph^1) \quad \text{for some} \quad \ell_0 \in  \N_0\,,
\end{equation}
where $\rho \in L^2((0,\infty); r \ud r)$ is a ($\lambda$-dependent) radial function.\footnote{In fact, the angular momentum of the kernel of $K_{T_c} + \lambda V$ must be \emph{even}, i.e.~$\ell_0 \in 2 \N_0$. This is because BCS theory is formulated for reflection symmetric $\alpha$, whence $K_{T_c}+ \lambda V$ is naturally defined on the space of reflection symmetric functions only. }

We can now formulate our main result in the case of pure angular momentum for $d=2$. 
(We again remark that, in the statement of \Cref{thm:main2}, one may replace $\Xi$ by the order parameter $|\Delta(\sqrt\mu)|$, see \Cref{remark.gap=order-parameter} above.)
\begin{thm}[{BCS energy gap for $2d$ pure angular momentum}] \label{thm:main2} 
  Let $d =2$, $\mu > 0$ and let $V$ satisfy Assumptions \ref{ass:basic} and \ref{ass:pureell}.
Define the critical temperature $T_c$ and energy gap $\Xi$ as in \eqref{eq:Tc} and \eqref{eq:energygap} with interaction $\lambda V$ for a $\lambda > 0$
and $\Delta$ being \emph{any~(arbitrary!)} non-zero solution the BCS gap equation \eqref{eq:gapeq} with interaction $\lambda V$. 

Then, with $\fBCS(h), h = \sqrt{1 - T/T_c}$ being the function defined via \eqref{eq:fBCS}, we have the following: 
\begin{itemize}
	\item[(a)] 
  Assume additionally that $V\in L^2(\R^2)$, $\hat V\in L^r(\R^2)$ for some $1\leq r < 2$ and that $|\cdot|^2 V \in L^1(\R^d)$.
  Then there exists $0 \leq \tilde{T} < T_c$ with $\tilde{T}/T_c \le \mathrm{e}^{-\mathfrak{c}/\lambda}$ for some $\mathfrak{c}>0$, 
  such that for all temperatures $T \in (\tilde{T}, T_c)$
  it holds that
	\begin{equation} \label{eq:maina2}
		\Xi(h) = T_c \, \fBCS(h) \, \big(1 + O(h^{-1} \ee^{-c/\lambda})\big)
	\end{equation}
	for some constant $c>0$ independent of $\lambda$ and $h$. 
	\item[(b)] Assuming additionally that $(1 + |\cdot|) V \in L^2(\R^d)$, it holds that
	\begin{equation} \label{eq:mainb2}
		\Xi(h) = T_c \, \fBCS(h) \, \big(1 + O(h \ee^{c'/\lambda}) + o_{\lambda \to 0}(1)\big)
	\end{equation}
	for some constant $c'>0$ independent of $\lambda$ and $h$ and where $o_{\lambda \to 0}(1)$ vanishes as $\lambda \to 0$ uniformly in $h$.
\end{itemize}
\end{thm}
The proof of Theorem \ref{thm:main2} is given in Section \ref{subsec:proof2d angmom}. 

\begin{remark}[On the assumptions] 
The additional assumptions in part (a) here compared to \Cref{thm:main} (namely $V\in L^2$ and $\hat V \in L^r$) are those of \cite[Theorem 2.1]{Deuchert.Geisinger.ea.2018}.
The proof of \Cref{eq:maina2} centrally uses this result. As discussed in \cite[Remark 2.3]{Deuchert.Geisinger.ea.2018}
these additional assumptions are expected to be of a technical nature.
\end{remark}

\begin{remark}[{The temperature $\tilde T$}] \label{rmk:extensionT}
The presence of the temperature $\tilde T$ in \Cref{thm:main2}~(a) arises from the first excited eigenvalue of $K_{T_c}+\lambda V$, 
see \cite[Remark 2.2]{Deuchert.Geisinger.ea.2018}. 
As discussed in the proof, the temperatures $T_c,\tilde T$ are given by  
$T_c=T_c(\ell_0)$ and $\tilde T = T_c(\ell_1)$, the critical temperatures restricted to angular momenta $\ell_0$ and $\ell_1$,
for some angular momenta $\ell_0\ne \ell_1$, 
see also \cite[Remark 2.2]{Deuchert.Geisinger.ea.2018}.
For temperatures $T \in (\tilde T, T_c)$ the BCS minimizer(s) then have angular momentum $\ell_0$ \cite[Theorem 2.1]{Deuchert.Geisinger.ea.2018}.
For temperatures $T < \tilde T$ however, we do not in general know whether the BCS minimizer(s) have angular momentum $\ell_0$. 
The proof crucially uses that the minimizer(s) have a definite angular momentum. 
If we however know a priori, that the BCS minimizer(s) have angular momentum $\ell_0$ for some larger ranger of temperatures $(T_1,T_c)$,
then the formula in \Cref{eq:maina2} holds in this larger range of temperatures. 
\end{remark}

\begin{remark}[Nodes of the gap function]
As already mentioned in \Cref{subsubsec:TcXiintro}, we establish during the proof, that any solution $\Delta$ of the BCS gap equation \eqref{eq:gapeq} has a radially symmetric
absolute value, $\abs{\Delta(p)}$, which is, moreover, independent of the particular solution $\Delta$. 
In particular, every solution $\Delta$ of the BCS gap equation \eqref{eq:gapeq} does \emph{not} have nodes on the Fermi surface. 
This contrasts many examples of $d$-wave superconductors in the physics literature, 
where a (necessarily) non-radial interaction $V$ leads to a gap function $\Delta$ with nodes on the Fermi surface, see, e.g., \cite{borkowski1994distinguishing, Maki.Won.1996, fletcher2009evidence, zhang2012nodal}. 
\end{remark}

\begin{remark}[{Non-extension to three dimensions}] \label{rmk:3dangmom}
The formula $\Xi(h) \approx T_c \fBCS{}(h)$ is \emph{not} expected to hold in three dimensions for non-zero angular momentum,
see for instance \cite[Figure~14.6]{Poole.Farach.ea.2007}. More precisely, we have the following: 
\begin{itemize}
	\item[(i)] For non-zero angular momentum in three dimensions, our method of proving Theorem~\ref{thm:main2}~(a) breaks down. In fact, we crucially use that $K_T^{\Delta} + \lambda V \geq 0$ for $\Delta = -2\lambda \widehat{V\alpha}$ with $\alpha$ a minimizer of the BCS functional. However, as shown in \cite[Proposition 2.11]{Deuchert.Geisinger.ea.2018} this implies that $\abs{\hat \alpha}$ is a radial function. In particular, in three dimensions, it cannot have a definite non-zero angular momentum.
	\item[(ii)] \emph{Assume} that we know \emph{a priori} that a solution of the BCS gap equation \eqref{eq:gapeq} (in spherical coordinates) satisfies $\Delta(p, \omega) = \Delta_0(p) Y_{\ell}^m(\omega)$ with $Y_\ell^m$ being the usual $L^2$-normalized (complex) spherical harmonic with $\ell \in \N_0$ and $m \in \{-\ell, ... , \ell\}$. Then, by application of \cite[Theorem 2.10]{Frank.Lemm.2016}, following very similar arguments to Sections \ref{subsec:GLproof} and \ref{subsubsec:partb}, we find that the radial part of the gap function is given by 
	\begin{equation} \label{eq:modCuniv}
|\Delta_0(\sqrt{\mu})| \approx c_{\ell,m} C_{\rm univ} h T_c 
	\end{equation}
	on the Fermi sphere $\{ p^2 = \mu\}$. Here $C_{\rm univ}$ was defined in \eqref{eq:Cuniv} and we denoted
	\begin{equation} \label{eq:clm}
		\begin{split}
		c_{\ell,m} :=& \left(\int_{\Sph^2} |Y_{\ell}^m(\omega)|^4 \ud \omega\right)^{-1/2} \\
		=&  \left(\sum_{L =0}^{2 \ell} \frac{(2 \ell+1)^2}{4 \pi (2 L +1)} |\langle \ell, \ell ; 0,0 \vert L;0 \rangle|^2 \, |\langle \ell, \ell ; m,m \vert L;2m \rangle|^2\right)^{-1/2}
		\end{split}
	\end{equation}
with $\langle \ell_1, \ell_2; m_1, m_2 | L;M \rangle$ being the well tabulated Clebsch-Gordan coefficients (see, e.g.,~\cite[p.~1046]{cohentanoudj}). The relation \eqref{eq:clm} shows that, in particular, even in the subspace of fixed angular momentum $\ell \neq 0$, the behavior \eqref{eq:modCuniv} is non-universal due to a non-trivial dependence on $m \in \{-\ell, ..., \ell \}$, as, for example (see \cite[Eq.~(6.8)]{Frank.Lemm.2016}), 
$$
c_{2,0} = \sqrt{\frac{28 \pi}{15}} \quad \text{and} \quad c_{2, \pm 1} = c_{2, \pm 2} = \sqrt{\frac{14 \pi}{5}}\,. 
$$

For temperatures $0 \le T \le T_c$ and $h:= \sqrt{1 - T/T_c}$, we expect \eqref{eq:modCuniv} to generalize to 
\begin{equation*}
|\Delta_0(\sqrt{\mu})| \approx T_c \, \fBCS^{(\ell, m)}(h)
\end{equation*}
with $\fBCS^{(\ell, m)} : [0,1] \to [0, \infty)$ being implicitly defined via
\begin{equation} \label{eq:fBCSangmom}
\int_{0}^\infty \ud s \int_{\Sph^{2}} \ud \omega \left\{\frac{\tanh \frac{\sqrt{s^2 + \big(\fBCS^{(\ell,m)}(h)\big)^2  |Y_{\ell}^m(\omega)|^2}}{2(1-h^2)}}{\sqrt{s^2 + \big(\fBCS^{(\ell,m)}(h)\big)^2 |Y_{\ell}^m(\omega)|^2}}   
-  \frac{\tanh \frac{s}{2}}{s} \right\}  |Y_{\ell}^m(\omega)|^2 = 0\,, 
\end{equation}
similarly to \cite[Eq.~(14.33)]{Poole.Farach.ea.2007}. For $\ell = m= 0$, \eqref{eq:fBCSangmom} yields that $\fBCS^{(0,0)} = (4 \pi)^{1/2} \fBCS$ with $\fBCS$ from \eqref{eq:fBCS} due to the $L^2$-normalization of the spherical harmonics (recall $\Delta(p, \omega) = \Delta_0(p) Y_{\ell}^m(\omega)$). 
\end{itemize}
A detailed analysis of the three-dimensional case with non-zero angular momentum is deferred to future work. 
\end{remark}

\section{Proofs of the main results} \label{sec:proofs}
This section contains the proofs of our main results formulated in Section \ref{sec:results}.

\subsection{Proof of Proposition \ref{prop:BCSoriguniv}} \label{subsec:proofBCSorig}
{For ease of notation, we shall henceforth write $\lambda$ instead of $\lambda_{\rm BCS}$.}
From the explicit form \eqref{eq:BCSgap} it is clear that $\Xi = \Delta$ and $\delta(h) \equiv \delta = \Delta/T_c$ is determined through \eqref{eq:diffformula}. Hence, the goal is to show that $\delta(h)/\fBCS(h) = 1 + O(\mathrm{e}^{-1/\lambda})$ uniformly in $h \in [0,1]$. The proof of this is conducted in three steps. 

\subsubsection{A priori bound on \texorpdfstring{$\delta$}{delta}}
We shall prove the following lemma. 

\begin{lemma} \label{lem:delta(h)apriori}
For $\delta = \delta(h)$ defined through \eqref{eq:diffformula} and $\lambda > 0$ small enough, it holds that 
\begin{equation} \label{eq:delta(h)apriori}
\delta(h) \le C h\,. 
\end{equation}
\end{lemma}
\begin{proof}
First, we note that $\delta(h) \le C$ uniformly for $h \in [0,1]$. This easily follows from observing that $\delta(h)$ is strictly monotonically increasing (as follows from elementary monotonicity properties of the integrand in \eqref{eq:diffformula}) and $\delta(1)$ is necessarily bounded. 

In order to show \eqref{eq:delta(h)apriori}, we employ the implicit function theorem to derive an asymptotic ODE for $\delta(h)$. For this purpose, we now introduce the function (recalling $T_c \sim \mathrm{e}^{-1/\lambda}$)
\begin{equation*} 
		G_\lambda : [0,1] \times [0,\infty) \to \R\,, (h, \delta) \mapsto 	\int_{0}^{T_D/T_c} \left\{ \frac{\tanh\big(\frac{\sqrt{x^2 + \delta^2}}{2(1 - h^2)}\big)}{\sqrt{x^2 + \delta^2}}  - \frac{\tanh\big(\frac{x}{2}\big)}{x} \right\} \mathrm{d}x  
\end{equation*}
and trivially note that \eqref{eq:diffformula} is equivalent to $G_\lambda(h, \delta(h)) = 0$. Since $G_\lambda$ is $C^1$ (away from the boundary) in $\delta $ and $h$ (this easily follows from dominated convergence), we can apply the implicit function theorem to obtain the differential equation 
\begin{equation} \label{eq:diffeqn}
	\frac{\partial \delta(h) }{\partial h}=   \frac{(1 - h^2) h }{\delta(h)} \left( \int_{0}^{T_D/T_c} \frac{1}{\cosh^2\left(\frac{\sqrt{x^2 + \delta^2}}{2 (1-h^2)}\right)} \D x \middle/ \int_{0}^{T_D/T_c} \frac{g_1\left(\frac{\sqrt{x^2 + \delta^2}}{1-h^2}\right)}{\frac{\sqrt{x^2 + \delta^2}}{ 1-h^2}} \D x\right)\,,
\end{equation}
where we introduced the auxiliary functions
\begin{equation}\label{eq:g01}
	g_0(z) := \frac{\tanh(z/2)}{z}\,, \qquad 
g_1(z) := - g_0'(z) = z^{-1} g_0(z) - \frac{1}{2} z^{-1} \frac{1}{\cosh^2(z/2)}\,. 
\end{equation}

It is elementary to check that the even function $z \mapsto g_1(z)/z$ is (strictly) positive and (strictly) decreasing for $z \in [0,\infty)$. In combination with $\delta(h) \le C$ and $T_c \sim \mathrm{e}^{-1/\lambda}$, one can thus bound the denominator on the r.h.s.~of \eqref{eq:diffeqn} from below. Together with an upper bound on the integral in the numerator (obtained by using elementary monotonicity properties of the hyperbolic cosine), we find that
\begin{equation} \label{eq:diffineqn}
	\frac{\partial \delta(h) }{\partial h} \le C' \,   \frac{h}{\delta(h)} \left( \int_{0}^{\infty} \frac{1}{\cosh^2(x)} \D x \middle/ \int_{0}^{C} \frac{g_1\left(\sqrt{x^2 + C^2}\right)}{\sqrt{x^2 + C^2}} \D x\right) \le C'' \,  \frac{h}{\delta(h)}
\end{equation}
for $h> 0$ and $\lambda > 0$ small enough (to ensure $T_D/T_c \ge C$). 

Finally, the differential inequality \eqref{eq:diffineqn} can be integrated using the boundary condition $\delta(0) = 0$ to conclude the desired.\footnote{Strictly speaking, this requires to extend the function $\delta(h)$ in $(0,1)$, obtained via the implicit function theorem for $G_\lambda$, to the boundary points $0$. In order to do so, note that, for $h \in (0,1/2)$, \eqref{eq:diffeqn} yields 
\begin{equation*}
\frac{\partial \delta(h)}{\partial h} \sim \frac{h}{\delta(h)}\,,
	\end{equation*}
from which we immediately conclude that $|\partial_h \delta(h)| \le C$, uniformly in $(0,1/2)$. Hence, $\delta(h)$ continuously extends to $0$. The same is true for its derivative by means of \eqref{eq:diffeqn} again. We remark that by a similar argument, $\delta(h)$ can be extended to $1$ as well.}
\end{proof}

\subsubsection{Uniform error estimate}
Having Lemma \ref{lem:delta(h)apriori} as an input, we shall now prove the following. 
\begin{lemma} \label{lem:errorunif}
For $\delta = \delta(h)$ defined through \eqref{eq:diffformula}, it holds that 
\begin{equation} \label{eq:errorunif}
\int_{T_D/T_c}^{\infty} \left|\frac{\tanh\big(\frac{\sqrt{x^2 + \delta^2}}{2(1 - h^2)}\big)}{\sqrt{x^2 + \delta^2}}  - \frac{\tanh\big(\frac{x}{2}\big)}{x}  \right| \mathrm{d}x \le C \, h^2 \, \mathrm{e}^{-2/\lambda}\,. 
\end{equation}
\end{lemma}
\begin{proof}
First, we add and subtract $\tanh(x/2)/\sqrt{x^2 + \delta^2}$ in \eqref{eq:errorunif}. Then, we employ $T_c \sim \mathrm{e}^{-1/\lambda}$ and Lemma~\ref{lem:delta(h)apriori} to estimate
\begin{equation*}
\int_{T_D/T_c}^{\infty} \left|\frac{\tanh\big(\frac{\sqrt{x^2 + \delta^2}}{2(1 - h^2)}\big)}{\sqrt{x^2 + \delta^2}}  - \frac{\tanh\big(\frac{x}{2}\big)}{\sqrt{x^2 + \delta^2}}  \right| \mathrm{d}x \le C \, h^2 \int_{T_D/T_c}^{\infty} \frac{1}{\cosh^2(x/2)} \mathrm{d}x \le C\, h^2 \, \mathrm{e}^{-2/\lambda}
\end{equation*}
and 
\begin{equation*}
\int_{T_D/T_c}^{\infty} \left|\frac{\tanh\big(\frac{x}{2}\big)}{\sqrt{x^2 + \delta^2}}  - \frac{\tanh\big(\frac{x}{2}\big)}{x}  \right| \mathrm{d}x \le C  \, h^2 \, \int_{T_D/T_c}^{\infty} \frac{1}{x^3} \mathrm{d}x \le C \, h^2 \, \mathrm{e}^{-2/\lambda}\,. 
\end{equation*}
Combining these bounds yields the claim by means of the triangle inequality. 
\end{proof}
From Lemma \ref{lem:errorunif}, we immediately conclude that 
\begin{equation} \label{eq:almostfBCS}
	\int_{\R} \left\{ \frac{\tanh\big(\frac{\sqrt{x^2 + \delta^2}}{2(1 - h^2)}\big)}{\sqrt{x^2 + \delta^2}}  - \frac{\tanh\big(\frac{x}{2}\big)}{x} \right\} \mathrm{d}x  = O(h^2 \, \mathrm{e}^{-2/\lambda})\,. 
\end{equation}
\subsubsection{Comparison with \texorpdfstring{$\fBCS{}$}{fBCS}}
Given \eqref{eq:almostfBCS}, the remaining task is to show that, because $\delta$ approximately \emph{solves} the defining equation of $\fBCS$, it is \emph{actually} close to $\fBCS$. This is the content of the following lemma.  
\begin{lemma} \label{lem:fBCSapprox}
Fix $h \in [0,1]$. If $\phi \in [0, \infty)$ satisfies\footnote{We follow the convention that, if $h=1$, we replace $\tanh(...)$ by the constant function $1$. }
\begin{equation} \label{eq:fBCSdefwitherr}
	\int_{\R} \left\{ \frac{\tanh\big(\frac{\sqrt{x^2 + \phi^2}}{2(1 - h^2)}\big)}{\sqrt{x^2 + \phi^2}}  - \frac{\tanh\big(\frac{x}{2}\big)}{x} \right\} \mathrm{d}x  = R_U
\end{equation}
for some $|R_U| \le C$, then 
\begin{equation} \label{eq:fBCSapprox}
\phi = \fBCS(h) + O(|R_U|^{1/2})
\end{equation}
with $\fBCS$ defined in \eqref{eq:fBCS}. 
\end{lemma}
Hence, combining \eqref{eq:almostfBCS} with \eqref{eq:fBCSapprox} and invoking Lemma \ref{lem:fBCS}~(iii), we find that 
\begin{equation*}
\delta(h) = \fBCS(h) + O(h \mathrm{e}^{-1/\lambda}) = \fBCS(h) \, \big(1 + O(\mathrm{e}^{-1/\lambda})\big)\,. 
\end{equation*}
This concludes the proof of Proposition \ref{prop:BCSoriguniv}. 
\begin{proof}[Proof of Lemma \ref{lem:fBCSapprox}]
First, we note that, given $h \in [0,1]$, \eqref{eq:fBCSdefwitherr} has a solution $\phi \in [0, \infty)$ if and only if 
$R_U \le \log\left(1/(1-h^2)\right)$. Then, as in the proof of \cite[Lemma 6]{Langmann.Triola.2023} (see \cite[Equation (C50)]{Langmann.Triola.2023}) we find that 
\begin{equation} \label{eq:fBCSwitherror}
	\phi = \mathrm{e}^{-R_U} \fBCS\left(\sqrt{h^2 + (1-h^2)(1-e^{R_U})}\right)\,.
\end{equation}
Taylor expanding in $R_U$ around $0$, using regularity of $\fBCS$ from Lemma \ref{lem:fBCS}, we get 
\begin{equation*}
	\mathrm{e}^{R_U}\phi
	= \fBCS(h) + \int_0^1 \fBCS'\left(\sqrt{h^2 + (1-h^2)(1-\mathrm{e}^{tR_U})}\right) \frac{-(1-h^2)R_U \mathrm{e}^{tR_U}}{2\sqrt{h^2 + (1-h^2)(1-\mathrm{e}^{tR_U})}} \ud t\,.
\end{equation*}
To bound the integral we change variables to $s = h^2 + (1-h^2)(1-e^{tR_U})$ and bound $| \fBCS'(h)| \leq C$ by \Cref{lem:fBCS}.
We split into cases depending on the sign of $R_U$. \\[1mm]
\underline{$R_U > 0$:}
For $R_U > 0$ the integral is bounded by 
\begin{equation*}
	\int_{h^2 - (1-h^2)(\mathrm{e}^{R_U}-1)}^{h^2} \frac{\ud s }{2\sqrt s}
	= h - \sqrt{h^2 - (1-h^2)(\mathrm{e}^{R_U}-1)} 
	= \frac{(1-h^2)(\mathrm{e}^{R_U}-1)}{\sqrt{h^2 - (1-h^2)(\mathrm{e}^{R_U}-1)} + h}.
\end{equation*}
Noting that $\frac{\delta}{\sqrt{\eps-\delta} + \sqrt\eps} \leq \sqrt\delta$ and that 
$(1-h^2)(\mathrm{e}^{R_U}-1) \leq C R_U$ we find that the integral is bounded by $\sqrt{R_U}$.\\[1mm]
\underline{$R_U < 0$:}
For $R_U < 0$ the integral is similarly bounded by 
\begin{equation*}
	\int_{h^2}^{h^2 + (1-h^2)(1-\mathrm{e}^{R_U})} \frac{\ud s }{2\sqrt s}
	= \sqrt{h^2 + (1-h^2)(1-\mathrm{e}^{R_U})} - h 
	\leq C \sqrt{|R_U|}\,. 
\end{equation*}

Plugging these bounds into \eqref{eq:fBCSwitherror}, we conclude the desired. 
\end{proof}

\subsection{Proof of Theorem \ref{thm:main}(a)} \label{subsec:partaproof}
We give here the proof of \Cref{thm:main}(a). 
The argument is divided into several steps.

\subsubsection{A priori spectral information}
For any temperature $T$ we have by \Cref{prop.KTDelta.geq.0} 
that there exists a unique (up to a constant global phase) minimizer $(\gamma_*, \alpha_*)$
of the BCS functional. The function $\alpha_*$ is radial and has $\hat\alpha_* > 0$. 
Moreover, the operator
$K_T^{\Delta} + \lambda V$ has lowest eigenvalue $0$ and $\alpha_*$ is the unique eigenfunction with this eigenvalue.

\subsubsection{Weak a priori bound on \texorpdfstring{$\Delta$}{Delta}}
From the proof of \Cref{prop:exofmin} in \cite{hainzl.seiringer.16} 
we have the following bound on the minimising $(\gamma_*, \alpha_*)$ of the BCS functional
\cite[Eqn. 3.12]{hainzl.seiringer.16}
\begin{multline*}
  \int (1 + p^2) (\abs{\hat \alpha_*(p)}^2 + \hat\gamma_*(p)) \ud p 
    \\
    \leq 8T \int \log \left(1 + e^{-(p^2-\mu)/4T}\right) \ud p
    + 8 \int \left[p^2/4 - 1 + C_2(\lambda)\right]_{-} \ud p
     \leq
    C_\lambda 
\end{multline*}
with $C_2(\lambda) = \inf \spec ( p^2/4 + \lambda V) \leq 0$
and thus $C_\lambda$ uniformly bounded for $\lambda$ small enough.
In particular $\norm{\alpha_*}_{H^1}\leq C$ uniformly for $\lambda$ small enough. 
By Sobolev's inequality \cite[Thm 8.3]{analysis} we then have 
\begin{align*}
   \norm{\alpha}_{L^\infty}^2
   & \leq C \norm{\nabla\alpha}_{L^2} \norm{\alpha}_{L^2}
   \leq C 
   \tag{$d=1$}
   \\
   \norm{\alpha}_{L^q}^2
   & \leq 
   C \norm{\nabla \alpha}_{L^2} \norm{\alpha}_{L^2}
   \leq C 
   \tag{$d=2$}
   \\
   \norm{\alpha}_{L^6}^2 
   & \leq C \norm{\nabla \alpha}_{L^2}^2 
   \leq C
   \tag{$d=3$}
  \end{align*} 
for any $2\leq q < \infty$.
Thus, 
\begin{equation}\label{eqn.weak.Delta.infty.a.priori}
  \norm{\Delta}_{L^\infty} 
  \leq 2\lambda \norm{V\alpha}_{L^1}
  \leq 
  \begin{cases}
  2 \lambda \norm{V}_{L^1} \norm{\alpha}_{L^\infty} 
  & d=1
  \\
  2 \lambda \norm{V}_{L^p} \norm{\alpha}_{L^2}^{(2pq-2q-2p)/(pq-2p)} \norm{\alpha}_{L^q}^{(2q-pq)/(pq-2p)}
  & d=2
  \\
  2\lambda \norm{V}_{L^{3/2}} \norm{\alpha}_{L^2}^{1/2} \norm{\alpha}_{L^6}^{1/2}
  & d=3
  \end{cases}
  \quad 
  \leq C \lambda 
\end{equation}
uniformly for $\lambda$ small enough (for any $1 < p < \min\{2,p_V\}$ by choosing $q$ large enough in dimension $d=2$).
In particular we see that $\Delta(p) \to 0$ as $\lambda \to 0$.

\subsubsection{Birman--Schwinger principle}\label{sec.Birman--Schwinger.(a)}
Next, by the Birman--Schwinger principle \cite{frank.hainzl.naboko.seiringer,hainzl.hamza.seiringer.solovej,hainzl.seiringer.16}
the fact that $K_T^\Delta + \lambda V$ has lowest eigenvalue $0$ with $\alpha_*$ being the unique eigenvector 
is equivalent to the Birman--Schwinger operator 
\begin{equation*}
B_{T,\Delta} : = \lambda V^{1/2} \frac{1}{K_T^\Delta} |V|^{1/2}
\end{equation*}
having $-1$ as its lowest eigenvalue and 
$\phi = V^{1/2}\alpha_*$ being the unique eigenfunction of this eigenvalue.
Here we use the convention that $V^{1/2}(x) = \sgn(V(x)) |V(x)|^{1/2}$.
We decompose the Birman--Schwinger operator into a dominant singular term and an error term.
For this purpose we define the (rescaled) Fourier transform restricted to the sphere
$\mfF_\mu : L^1(\R^d) \to L^2(\S^{d-1})$ by 
\begin{equation*}
\mfF_\mu \psi (p) = \frac{1}{(2\pi)^{d/2}} \int_{\R^d} \psi(x) e^{-\mathrm{i}\sqrt\mu p\cdot x} \ud x,
\end{equation*}
which is well-defined by the Riemann--Lebesgue Lemma. Define then  
\begin{equation}\label{eqn.define.m(T,Delta)}
  m(T,\Delta) = \frac{1}{\abs{\S^{d-1}}} \int_{|p|\leq \sqrt{2\mu}} \frac{1}{K_T^\Delta(p)} \ud p
\end{equation}
and decompose
\begin{equation*}
 B_{T,\Delta} = \lambda m(T,\Delta) V^{1/2} \mfF_\mu^\dagger \mfF_\mu |V|^{1/2} + \lambda V^{1/2} M_{T,\Delta} |V|^{1/2},
\end{equation*}
with $M_{T,\Delta}$ defined such that this holds.
Analogously to \cite[Lemma 2]{frank.hainzl.naboko.seiringer} and \cite[Lemma 3.4]{Henheik.Lauritsen.ea.2023} we have the following lemma, the proof of which (that is analogous to the one of \cite[Lemma 3.4]{Henheik.Lauritsen.ea.2023}) is given in \Cref{sec.proof.lem.bdd.VMV}.
\begin{lemma}
\label{lem.bdd.VMV}
We have 
$\norm{V^{1/2} M_{T,\Delta} |V|^{1/2}}_{\textnormal{HS}} \leq C$ for small $\lambda$ uniformly in $T$ and $\Delta$, where $\Vert \cdot \Vert_{\mathrm{HS}}$ denotes the Hilbert-Schmidt norm of an operator.
\end{lemma}

We conclude that $1 + \lambda V^{1/2} M_{T,\Delta} |V|^{1/2}$ is invertible for sufficiently small $\lambda$ and thus,
analogously to \cite[Lemma 4]{Hainzl.Seiringer.2008} and \cite[Lemma 13, Proposition 15]{Henheik.Lauritsen.2022}
the fact that $B_{T,\Delta}$ has lowest eigenvalue $-1$ is equivalent to the fact that the operator
\begin{equation}
\label{eqn.T.T.Delta.def}
  S_{T,\Delta} := \lambda m(T,\Delta) \mfF_\mu |V|^{1/2} \frac{1}{1 + \lambda V^{1/2} M_{T,\Delta} |V|^{1/2}} V^{1/2} \mfF_\mu^{\dagger} 
\end{equation}
acting on $L^2(\S^{d-1})$ has lowest eigenvalue $-1$.
Moreover, the function $\mfF_\mu |V|^{1/2} \phi = \mfF_\mu V \alpha_*$ is the unique eigenfunction of $S_{T,\Delta}$ of eigenvalue $-1$.

\subsubsection{A priori bounds on \texorpdfstring{$\Delta$}{Delta}}\label{sec.a.priori.Delta.(a)}
For the analysis of the integral $m(T,\Delta)$ we need some a priori bounds on $\Delta$. 
Analogously as in \cite{Henheik.Lauritsen.ea.2023,Hainzl.Seiringer.2008} we need some control of $\Delta(p)$ in terms of $\Delta(\sqrt\mu)$ 
and some type of Lipschitz-continuity of $\Delta$. 
These are collected in the following lemma. 
\begin{lemma}\label{lem.Delta.a.priori.bdd}
The function $\Delta$ satisfies the bounds
\begin{align}
  \Delta(p) 
  & 
  \leq C \Delta(\sqrt{\mu}),
  \label{eqn.Delta.bdd.Delta.sqrt.mu}
  \\
  \abs{\Delta(p) - \Delta(q)}
  &  
  \leq C \Delta(\sqrt{\mu}) |p-q| 
  \label{eqn.Delta.lipschitz}
\end{align}
for sufficiently small $\lambda$.
\end{lemma}

\begin{proof}
As noted above, the function $\mfF_\mu |V|^{1/2} \phi = \mfF_\mu V\alpha$ is the eigenfunction of $S_{T,\Delta}$
of lowest eigenvalue $-1$.

Further, to leading order, $S_{T,\Delta}$ is proportional to $\mcV_\mu = \mfF_\mu V \mfF_\mu ^\dagger$.
Since the constant function $u = \abs{\S^{d-1}}^{-1/2} \in L^2(\S^{d-1})$ is the 
ground state of $\mathcal{V}_\mu$ (see the argument around \eqref{eq:emu}), the same is also true for $S_{T,\Delta}$ whenever $\lambda$ is small enough.
Hence, one can easily see that
\begin{equation*}
\frac{1}{1 + \lambda V^{1/2}M_{T,\Delta}|V|^{1/2}} V^{1/2} \mfF_{\mu}^\dagger u
\end{equation*}
is an eigenvector of $B_{T,\Delta}$ with eigenvalue $-1$ and thus proportional to $\phi = V^{1/2} \alpha_*$.
Thus, with $\mfF$ denoting the usual Fourier transform, by expanding 
$\frac{1}{1+x} = 1 - \frac{x}{1+x}$ we have
\begin{align*}
\Delta 
  & = -2 \lambda \mfF |V|^{1/2} \phi
  = f(\lambda) \mfF |V|^{1/2} \frac{1}{1 + \lambda V^{1/2} M_{T,\Delta} |V|^{1/2}} V^{1/2} \mfF_\mu^\dagger u
  \\
  &  =
  f(\lambda) \left(\int_{\S^{d-1}} \hat V(p-\sqrt{\mu}q) \ud \omega(q) + \lambda \eta_\lambda(p)\right)
\end{align*} 
for some constant $f(\lambda)$ (where we absorbed the factor $|\S^{d-1}|^{-1/2}(2\pi)^{-d/2}$ into $f(\lambda)$).
One easily verifies that 
\begin{equation*}
\eta_\lambda = -(2\pi)^{d/2}\mfF |V|^{1/2} \frac{V^{1/2} M_{T,\Delta} |V|^{1/2}}{1 + \lambda V^{1/2} M_{T,\Delta} |V|^{1/2}} V^{1/2} \mfF_\mu^\dagger u
\end{equation*}
has $\norm{\eta_\lambda}_\infty \leq C$ uniformly in sufficiently small $\lambda$ by \Cref{lem.bdd.VMV}.
By evaluating at $p=\sqrt{\mu}$ we find $\abs{f(\lambda)} \leq C \Delta(\sqrt{\mu})$ for small $\lambda$ 
and thus the global bound \eqref{eqn.Delta.bdd.Delta.sqrt.mu}.
Moreover, we have the Lipschitz-bound:
\begin{align*}
\abs{\Delta(p) - \Delta(q)}
  & \leq |f(\lambda)|
  |p-q|
  \norm{|x| |V|^{1/2} \frac{1}{1+ \lambda V^{1/2} M_{T,\Delta} |V|^{1/2}} V^{1/2} \mfF_\mu^\dagger u}_{L^1}
  \nonumber
  \\
  & \leq 
  C \Delta(\sqrt{\mu})
  |p-q|
  \norm{|x|^2 V}_{L^1}^{1/2} \frac{1}{1 - \lambda \norm{V^{1/2} M_{T,\Delta} |V|^{1/2}}} \norm{V^{1/2} \mfF_\mu^\dagger u}_{L^2}
  \nonumber
  \\ &
  \leq C \Delta(\sqrt{\mu}) |p-q|
\end{align*}
for sufficiently small $\lambda$.
\end{proof}

\subsubsection{First order}\label{sec.first.order.(a)}
Expanding the geometric series in \Cref{eqn.T.T.Delta.def} to first order we see that $S_{T,\Delta}$ to leading order is proportional 
to the operator $\mcV_\mu$ (defined in \eqref{eq:Vmu} above). Moreover, we have
\begin{align*}
  -1 & = \lambda m(T,\Delta) \inf \spec \left( \mfF_\mu |V|^{1/2} \frac{1}{1 + \lambda V^{1/2} M_{T,\Delta} |V|^{1/2}} V^{1/2} \mfF_\mu^{\dagger} \right)
    \\ & = \lambda m(T,\Delta) \inf \spec \mcV_\mu ( 1 + O(\lambda)) 
    = \lambda   e_\mu m(T,\Delta) ( 1 + O(\lambda)).
\end{align*}
In particular $m(T,\Delta) = \frac{-1}{\lambda e_\mu} (1 + O(\lambda)) \to \infty$ as $\lambda \to 0$.

\subsubsection{Exponential vanishing of \texorpdfstring{$\Delta$}{Delta}}
\label{sec.exp.vanish.Delta}
Pointwise we may bound
$K_T^\Delta \geq E_\Delta$. Thus, by the first--order analysis above, we have 
\begin{equation*}
  \frac{-1}{\lambda e_\mu} (1+O(\lambda)) 
  = m(T,\Delta) 
  = \frac{1}{\abs{\S^{d-1}}} 
  \int_{|p|< \sqrt{2\mu}} \frac{1}{K_T^\Delta} \ud p 
  \leq \frac{1}{|\S^{d-1}|} 
  \int_{|p| < \sqrt{2\mu}} \frac{1}{\sqrt{|p^2-\mu|^2 + |\Delta(p)|^2}} \ud p
\end{equation*}
The latter integral is calculated in \cite{Hainzl.Seiringer.2008,Henheik.Lauritsen.ea.2023}.
The same calculation is valid here by the bounds \eqref{eqn.Delta.bdd.Delta.sqrt.mu} and \eqref{eqn.Delta.lipschitz} and 
the fact that $\Delta(p)\to 0$ by \eqref{eqn.weak.Delta.infty.a.priori}. 
That is 
\begin{equation*}
  \frac{-1}{\lambda e_\mu} 
  \leq \mu^{d/2-1} \left(\log \frac{\mu}{\Delta(\sqrt{\mu})} + O(1)\right)
\end{equation*}
in the limit $\lambda \to 0$. 
We conclude that 
$\Delta(\sqrt\mu) \lesssim e^{-c/\lambda}$ as $\lambda \to 0$ (with $c = -1/e_\mu\mu^{d/2-1}$). 

The constant $c$ shall henceforth be used generically and its precise value might change from line to line.

\subsubsection{Infinite order}\label{sec.infinite.order.(a)}
Recall that the unique eigenfunction of $S_{T,\Delta}$ of eigenvalue $-1$ is for small $\lambda$ given by the constant function $u$.
Thus, for small $\lambda$ we have 
\begin{equation*}
  -1 = \lambda m(T,\Delta) 
    \longip{u}{ \mfF_\mu |V|^{1/2} \frac{1}{1 + \lambda V^{1/2} M_{T,\Delta} |V|^{1/2} } V^{1/2} \mfF_\mu^\dagger}{u}.
\end{equation*}
Combining this for the temperatures $T$ and $T_c$ we find 
\begin{align}
  & m(T,\Delta) - m(T_c,0)
  \nonumber
  \\ & \quad 
    = \frac{-1}{\lambda} 
      \left[
      \frac{1}{\longip{u}{ \mfF_\mu |V|^{1/2} \frac{1}{1 + \lambda V^{1/2} M_{T,\Delta} |V|^{1/2} } V^{1/2} \mfF_\mu^\dagger}{u}}
      \right.
      -
      \left.
      \frac{1}{\longip{u}{ \mfF_\mu |V|^{1/2} \frac{1}{1 + \lambda V^{1/2} M_{T_c,0} |V|^{1/2} } V^{1/2} \mfF_\mu^\dagger}{u}}
      \right]
      \nonumber
    \\
    & \quad 
    = \frac{-1}{\lambda e_\mu^2}(1 + O(\lambda))
      \longip{u}{ \mfF_\mu |V|^{1/2} 
        \left(
          \frac{1}{1 + \lambda V^{1/2} M_{T,\Delta} |V|^{1/2} } 
          - 
          \frac{1}{1 + \lambda V^{1/2} M_{T_c,0} |V|^{1/2} }
        \right)
          V^{1/2} \mfF_\mu^\dagger}{u}
\label{eqn.diff.m(T.Delta).m(Tc)}
\end{align}
for small enough $\lambda$ by expanding to first order in the denominator and noting that $\inf \spec \mcV_\mu = e_\mu$.
The proof of the following lemma (which is somewhat analogous to the proof of \Cref{lem.bdd.VMV}) is given in \Cref{sec.proof.lem.bdd.V.MTDelta.minus.MTc.V}.
\begin{lemma}
Uniformly in small $\lambda$ we have the bound 
\label{lem.bdd.V.MTDelta.minus.MTc.V}
\begin{equation*}
\norm{V^{1/2}(M_{T,\Delta} - M_{T_c,0})|V|^{1/2}}_{\textnormal{HS}} 
\leq C e^{-c/\lambda}
\end{equation*}
for some constants $C,c > 0$.
\end{lemma}

Having \Cref{lem.bdd.V.MTDelta.minus.MTc.V} at hand, we write the difference as a telescoping sum
\begin{align*}
  & \frac{1}{1 + \lambda V^{1/2} M_{T,\Delta} |V|^{1/2} } 
  - 
  \frac{1}{1 + \lambda V^{1/2} M_{T_c,0} |V|^{1/2} }
  \\ & 
  \quad 
  = \sum_{k=1}^\infty (-\lambda)^k 
  \left[\left(V^{1/2} M_{T,\Delta} |V|^{1/2}\right)^k
      -
      \left(V^{1/2}M_{T_c,0}|V|^{1/2}\right)^k\right]
  \\ & 
  \quad 
  = 
  \sum_{k = 1}^\infty 
  (-\lambda)^k 
  \sum_{\ell = 0}^{k-1}
  \left(V^{1/2} M_{T,\Delta} |V|^{1/2}\right)^{k-1-\ell} 
  V^{1/2}(M_{T,\Delta} - M_{T_c,0})|V|^{1/2}
  \left(V^{1/2}M_{T_c,0}|V|^{1/2}\right)^{\ell}.
\end{align*}
Thus, by \Cref{lem.bdd.VMV,lem.bdd.V.MTDelta.minus.MTc.V} we have
\begin{align*}
\norm{\frac{1}{1 + \lambda V^{1/2} M_{T,\Delta} |V|^{1/2} } 
  - 
  \frac{1}{1 + \lambda V^{1/2} M_{T_c,0} |V|^{1/2} }}_{\textnormal{HS}}
  & \leq \sum_{k=1}^\infty \lambda^k \sum_{\ell=0}^{k-1} 
  C^{k-1-\ell} \times Ce^{-c/\lambda} \times C^\ell
  \\
  & \leq \sum_{k=1}^\infty k \lambda^k C^k e^{-c/\lambda} 
  \leq C \lambda e^{-c/\lambda}.
\end{align*}
We conclude that 
\begin{equation}\label{eqn.m.difference.infinite.order}
  \abs{m(T,\Delta) - m(T_c,0)}
  \leq C e^{-c/\lambda}.
\end{equation}

\subsubsection{Calculation of the integral \texorpdfstring{$m(T,\Delta) - m(T_c,0)$}{m(T,Delta)-m(Tc,0)}}
\label{sec.calc.int.m-m}
To extract the asymptotics in \eqref{eq:maina} from the bound in \eqref{eqn.m.difference.infinite.order}
we calculate the difference $m(T,\Delta) - m(T_c,0)$ and show that it is essentially the left-hand-side of \eqref{eq:fBCS}.
The argument is essentially given in \cite[Appendix C.4]{Langmann.Triola.2023}. For completeness, we give the argument here.

By changing variables to $s = (p^2 - \mu)/\mu$ and defining $x(s) = \Delta(\sqrt{\mu(1+s)})/\mu$ 
we get 
\begin{align*}
m(T,\Delta) - m(T_c,0)
  & = 
  \int_0^{\sqrt{2\mu}} \left(\frac{1}{K_T^\Delta} - \frac{1}{K_{T_c}}\right) p^{d-1} \ud p
  \\
  & = \frac{\mu^{d/2-1}}{2}  \int_{-1}^1 
    \left(
      \frac{\tanh\left(\frac{\sqrt{s^2+x(s)^2}}{2T/\mu}\right)}{\sqrt{s^2 + x(s)^2}}
      -
      \frac{\tanh \frac{s}{2T_c/\mu}}{s} 
    \right)
      (1 + s)^{d/2-1} 
    \ud s.
\end{align*}
This is of the form where we can use \cite[Lemma 5]{Langmann.Triola.2023}.
\begin{lemma}[{\cite[Lemma 5]{Langmann.Triola.2023}}]
Let $g,G$ be functions with $g(0) = G(0) = 1$ and $g\in L^\infty$ 
and let $\tau, \tau_c, \delta > 0$. Assume that $\tilde g(s) := (g(s)-1)/s$ and 
$\tilde G(s) := (G(s)-1)/s$ satisfy $\tilde g, \tilde G \in L^\infty(\R)$.
Let $s_1 > 0$ such that $g(s) > 1/2$ for $|s| < s_1$ and define 
\begin{equation*}
\begin{aligned}
J_{\tau,\delta,\tau_c}(g,G)
  & = \int_\R \left\{\frac{\tanh \frac{\sqrt{s^2 + g(s)^2 \delta^2}}{2\tau}}{\sqrt{s^2 + g(s)^2 \delta^2}}  
      -  \frac{\tanh \frac{s}{2\tau_c}}{s} \right\} G(s)
      \ud s,
  \\ 
J_{\tau,\delta,\tau_c}^{(0)}
= J_{\tau,\delta, \tau_c}(1,1)
  & = \int_\R 
    \left\{\frac{\tanh \frac{\sqrt{s^2 + \delta^2}}{2\tau}}{\sqrt{s^2 + \delta^2}}  
      -  \frac{\tanh \frac{s}{2\tau_c}}{s} \right\} \ud s.
\end{aligned}
\end{equation*}
Then 
\begin{equation*}
\begin{aligned}
\abs{J_{\tau,\delta,\tau_c}(g,G) - J_{\tau,\delta,\tau_c}^{(0)}}
  & \leq \norm{\tilde G}_{L^\infty} \left(4\tau + 4\tau_c + \pi \delta \norm{g}_{L^\infty}\right)
    + 4\delta \norm{\tilde g}_{L^\infty} \left(1 + \norm{g}_{L^\infty}\right)\left(1 + \frac{\delta}{2s_1}\right).
\end{aligned}
\end{equation*}
\end{lemma}
To apply this lemma we write 
\begin{equation*}
x(s) = \frac{\Delta(\sqrt{\mu(1+s)})}{\Delta(\sqrt{\mu})} \frac{\Delta(\sqrt{\mu})}{\mu}
  = g(s) \delta.
\end{equation*}
Then $g\in L^\infty$ uniformly in $\lambda$ by \eqref{eqn.Delta.bdd.Delta.sqrt.mu}
and $\tilde g \in L^\infty$ uniformly in $\lambda$ by \eqref{eqn.Delta.lipschitz}.
Finally, clearly $G(s) = (1+s)^{d/2-1}\chi_{|s|\leq 1}$ has $\tilde G \in L^\infty$.
We conclude that 
\begin{equation*}
  m(T,\Delta) - m(T_c,0)
    = \frac{\mu^{d/2-1}}{2} 
      \int_\R \left\{\frac{\tanh \frac{\sqrt{s^2 + (\Delta(\sqrt{\mu})/\mu)^2}}{2T/\mu}}{\sqrt{s^2 + (\Delta(\sqrt{\mu})/\mu)^2}}  
      -  \frac{\tanh \frac{s}{2T_c/\mu}}{s} \right\} \ud s
    + O(e^{-c/\lambda}).
\end{equation*}
Writing $T = T_c (1 - h^2)$, recalling the bound in \eqref{eqn.m.difference.infinite.order} and changing variables
we find 
\begin{equation*}
  \int_\R \left\{\frac{\tanh \frac{\sqrt{s^2 + (\Delta(\sqrt{\mu})/T_c)^2}}{2(1-h^2)}}{\sqrt{s^2 + (\Delta(\sqrt{\mu})/T_c)^2}}  
      -  \frac{\tanh \frac{s}{2}}{s} \right\} \ud s
    = R_U,
\end{equation*}
with $R_U = O(e^{-c/\lambda})$. Hence, by Lemma \ref{lem:fBCSapprox}, we find that 
\begin{equation}\label{eqn.f=fBCS.(a)}
	\frac{\Delta(\sqrt{\mu})}{T_c} 
	= \fBCS(h) + O(|R_U|^{1/2}) = \fBCS(h) ( 1 + O(h^{-1} e^{-c/\lambda}))
\end{equation}
since $\fBCS(h) \sim h$ for small $h$ by Lemma \ref{lem:fBCS}.

\subsubsection{Comparing \texorpdfstring{$\Delta(\sqrt{\mu})$ and $\Xi$}{Delta(sqrt(mu)) and Xi}}
\label{sec.compare.Xi.Delta.(a)}
We finally prove that $\Xi$ is essentially given by $\Delta(\sqrt\mu)$.

Clearly $\Xi = \inf_p E_\Delta(p) \leq E_\Delta(\sqrt{\mu}) = \Delta(\sqrt\mu)$. 
To show a corresponding lower bound consider 
$p$ with $|p^2-\mu| \leq \Xi \leq \Delta(\sqrt\mu)$. 
Then by \Cref{eqn.Delta.lipschitz} and the bound $\Delta(\sqrt\mu) = O(e^{-c/\lambda})$ we have 
\begin{align*}
\Delta(p) & \geq \Delta(\sqrt\mu) - C \Delta(\sqrt\mu) |p-\sqrt\mu|
\geq \Delta(\sqrt\mu) \left(1 - C \Delta(\sqrt\mu)\right)
\geq \Delta(\sqrt\mu) ( 1 + O(e^{-c/\lambda})).
\end{align*}
We conclude that 
\begin{equation}\label{eqn.compare.Delta.Xi}
  \Xi = \Delta(\sqrt\mu) \left(1 + O(e^{-c/\lambda})\right).
\end{equation}
Together with \eqref{eqn.f=fBCS.(a)} this concludes the proof of \Cref{thm:main}(a).

\subsection{Non-universal property of \texorpdfstring{$\Delta$}{Delta}: Proof of \texorpdfstring{\Cref{eqn.nonuniversality}}{Equation (\ref*{eqn.nonuniversality})}} \label{subsec:nonuniv}

From the Birman--Schwinger argument (\Cref{sec.Birman--Schwinger.(a)}) we have that $\phi = V^{1/2}\alpha_*$ 
is the (unique) eigenvector of 
\begin{equation*}
\frac{\lambda m(T,\Delta)}{1 + \lambda V^{1/2} M_{T,\Delta}|V|^{1/2}} V^{1/2} \mfF_\mu^\dagger \mfF_\mu |V|^{1/2}
\end{equation*}
of eigenvalue $-1$. Recalling that $\Delta = -2\lambda \mfF V \alpha_*$ we thus get the equation 
\begin{equation*}
\Delta = - \mfF |V|^{1/2} \frac{\lambda m(T,\Delta)}{1 + \lambda V^{1/2}M_{T,\Delta}|V|^{1/2}}V^{1/2}\mfF_\mu^\dagger \Delta(\sqrt\mu\, \cdot)
\end{equation*}
with $\Delta(\sqrt\mu\,\cdot)$ being the constant function on the unit sphere of value $\Delta(\sqrt\mu)$.
Recall that $-1 = \lambda m(T,\Delta) \longip{u}{\mfF_\mu |V|^{1/2} \frac{1}{1 + \lambda V^{1/2}M_{T,\Delta}|V|^{1/2}}V^{1/2}\mfF_\mu^\dagger}{u}$ 
for small enough $\lambda$.
By the same argument as in \Cref{sec.infinite.order.(a)} we may replace $M_{T,\Delta}$ by $M_{0,0}$, 
its corresponding value at $T=\Delta=0$, up to errors of order $e^{-c/\lambda}$.
(Concretely one can define $M_{0,0}$ via the representation of its kernel as given in 
\Cref{eqn.decompose.VMV,eqn.kernel.Vp^2V,eqn.kernel.V(M-p^2)V,eqn.kernel.Vp2chiV,eqn.kernel.MTDelta.d=1.d=2}, 
setting $T=\Delta=0$.)
Hence, for sufficiently small $\lambda$,
\begin{equation*}
\frac{\Delta}{\Delta(\sqrt\mu)}
	= \frac{|\S^{d-1}|^{1/2}\mfF |V|^{1/2} \frac{1}{1 + \lambda V^{1/2}M_{0,0}|V|^{1/2}}V^{1/2}\mfF_\mu^\dagger u}
		{\longip{u}{\mfF_\mu |V|^{1/2} \frac{1}{1 + \lambda V^{1/2}M_{0,0}|V|^{1/2}}V^{1/2}\mfF_\mu^\dagger }{u}}
		+ O(e^{-c/\lambda})
	= F + O(e^{-c/\lambda}).
\end{equation*}
Clearly, the function $F$ does not depend on the temperature $T$.

\subsection{Ginzburg-Landau theory: Proof of Theorem \ref{thm:main}(b)} \label{subsec:GLproof}

As mentioned above, the proof of Theorem \ref{thm:main}(b) builds on Ginzburg--Landau (GL) theory. 
For convenience of the reader, we recall the main input from GL theory for the purpose of the present paper in Proposition \ref{prop:GLtheory} below. More general and detailed statements can be found in the original papers \cite{Frank.Hainzl.ea.2012, Frank.Lemm.2016, deuchert2023microscopic, deuchert2023microscopic2}. In particular, these works allow for external fields or ground state degeneracy (cf.~Lemma \ref{lem:KTcGS} below), respectively.

As a preparation for Proposition \ref{prop:GLtheory}, we have the following lemma.

\begin{lemma}[Ground state of $K_{T_c} + \lambda V$] \label{lem:KTcGS}
	Let $V$ satisfy Assumptions~\ref{ass:basic} and \ref{ass:hatV neg}.  Then $K_{T_c} + \lambda V$ has $0$ as a non-degenerate ground state eigenvalue and its $L^2(\R^d)$-normalized ground state $\mathfrak{a}_0$ can be chosen to have strictly positive Fourier transform. Moreover, it holds that $\hat{\mathfrak{a}}_0 \in L^\infty(\R^d)$. 
\end{lemma}
\begin{proof}
Since $T_c>0$ (recall the discussion below Assumption \ref{ass:hatV neg}), we first note that the Fourier multiplier $K_{T_c}$  is strictly positive. Then using $\hat{V} \le0$, the claim follows from a Perron-Frobenius type argument (see also \cite{Hainzl.Seiringer.2008} and \cite[Assumption~2]{Frank.Hainzl.ea.2012}). The fact that $\hat{\mathfrak{a}}_0 \in L^\infty(\R^d)$ follows from \cite[Proposition~2]{Frank.Hainzl.ea.2012} by invoking Assumption \ref{ass:basic}. 
\end{proof}

We can now formulate the main results from GL theory, needed for the present paper.

\begin{prop}[{Ginzburg-Landau theory, see \cite[Theorem 2.10]{Frank.Lemm.2016}}] \label{prop:GLtheory}
Let $V$ be a function satisfying Assumptions \ref{ass:basic} and \ref{ass:hatV neg} and suppose that $ 0 \le T< T_c$. 
Then, using the notations from Proposition \ref{prop.KTDelta.geq.0} and Lemma \ref{lem:KTcGS}, we have that
\begin{equation*}
\mathcal{F}_T[\Gamma_*] - \mathcal{F}_T[\Gamma_{\rm FD}] = h^4  \mathcal{E}_{\textnormal{GL}}(\psi_{\rm GL}) + \mathcal{O}(h^6) \quad \text{as} \quad h \to 0\,,
\end{equation*}
where $\psi_{\rm GL} \neq 0$ minimizes the Ginzburg-Landau ``functional'' $\mcE_{\textnormal{GL}} : \C \to \R$,
\begin{equation} 
	\begin{split}\label{eq:GLfunctional}
\mathcal{E}_{\rm GL}(\psi) = & |\psi|^4\left[\frac{1}{T_c^3}\int_{\R^d} \frac{g_1((p^2 - \mu)/T_c)}{(p^2 - \mu)/T_c} |K_{T_c}(p)|^4 |\hat{\mathfrak{a}}_0(p)|^4\D p\right]  \\ 
& \hspace{10mm}- |\psi|^2 \left[\frac{1}{2 T_c}\int_{\R^d} \frac{1}{\cosh^2\left( (p^2 - \mu) / (2 T_c)\right)} |K_{T_c}(p)|^2 |\hat{\mathfrak{a}}_0(p)|^2\D p \right] \,.
	\end{split}
\end{equation}
Here we used the auxiliary function $g_1$ from \eqref{eq:g01}. 
Moreover, we can decompose the off--diagonal element $\hat{\alpha}_*$ of $\Gamma_*$ as 
\begin{equation} \label{eq:approxmin}
\hat{\alpha}_* = h |\psi_0| \hat{\mathfrak{a}}_0 + \hat{\xi}
\end{equation}
where $\Vert \xi \Vert_{L^2} = \mathcal{O}(h^2)$ and $\psi_0\neq 0$ approximately minimizes \eqref{eq:GLfunctional}, i.e.
\begin{equation} \label{eq:approx minimum}
\mathcal{E}_{\rm GL}(\psi_0) \le \mathcal{E}_{\rm GL}(\psi_{\rm GL}) + \mathcal{O}(h^2) \,.
\end{equation}
\end{prop}

\begin{remark} \label{rmk:tcinv}
We emphasize that all error terms in the above proposition (and also the implicit constants hidden in $\hat{\mathfrak{a}}_0 \in L^\infty(\R^d)$) are \emph{not} uniform in $\lambda$. This crucially limits the applicability of our GL theory based method for temperatures slightly away from the critical one with $\mathrm{e}^{-c'/\lambda}\ll h\ll 1$  (cf.~the error bound in \eqref{eq:mainb}). Indeed, a careful examination of the proofs in \cite{Frank.Hainzl.ea.2012, Frank.Lemm.2016} reveals that the hidden dependencies on the critical temperature $T_c$ are at most inverse polynomially and hence exponential in $\lambda$, i.e.~$\ee^{c/\lambda}$ for some $c > 0$ (independent of $\lambda$ and $h$). 

\end{remark}

\subsubsection{Minimizing the Ginzburg-Landau functional} 
Given the inputs from GL theory, Theorem \ref{thm:main}(b) is based on the following Proposition~\ref{prop:GLmin}, 
the proof of which we postpone after finishing the proof of Theorem~\ref{thm:main}(b).

\begin{prop}\label{prop:GLmin}
	The (up to a phase unique) minimizer $\psi_{\rm GL}$ of the GL functional \eqref{eq:GLfunctional} satisfies 
	\begin{equation} \label{eq:psiGL}
|\psi_{\rm GL}| = C_{\rm univ} \,   \frac{T_c}{\Delta_0(\sqrt{\mu})} \big(1 + o_{\lambda \to 0}(1)\big) \,,
	\end{equation}
where $\Delta_0:= - 2 (2\pi)^{-d/2} \lambda \hat{V} \star \hat{\mathfrak{a}}_0$,\footnote{Note that $\Delta_0$ was denoted by $t$ in \cite{Frank.Hainzl.ea.2012, Frank.Lemm.2016}. We also remark that $\Delta_0(\sqrt{\mu}) \neq 0$ as follows from $\hat{\mathfrak{a}}_0 > 0$ (by Lemma \ref{lem:KTcGS}) and $\hat{V} \le 0$ (by Assumption \ref{ass:hatV neg}). } the constant $C_{\rm univ}$ is given in \eqref{eq:Cuniv}, and the error is uniform in $h$. 
\end{prop}

The fact that $|\psi_0|>0$ approximately minimizes \eqref{eq:GLfunctional} (see \eqref{eq:approx minimum}), implies that (recalling Remark \ref{rmk:tcinv})
\begin{equation*}
	|\psi_0| = |\psi_{\rm GL}| +O\big(h \ee^{c/\lambda}\big)\,. 
\end{equation*}
Therefore, by means of \eqref{eq:approxmin} in combination with Proposition \ref{prop:GLmin}, we infer
\begin{equation*}
	\hat{\alpha}_* = C_{\rm univ} \, T_c \, h \,  \frac{\hat{\mathfrak{a}}_0}{\Delta_0(\sqrt{\mu})} \big( 1 + o_{\lambda \to 0}(1) + O\big(h \ee^{c/\lambda}\big)\big)+ \hat{\xi}\,,
\end{equation*}
and thus, after taking the convolution with $\lambda \hat{V}$, 
\begin{equation} \label{eq:Delta exp}
	\Delta = C_{\rm univ} \, T_c \, h  \, \frac{\Delta_0}{\Delta_0(\sqrt{\mu})} \big( 1 + o_{\lambda \to 0}(1) + O\big(h \ee^{c/\lambda}\big)\big)+ \Delta_\xi\,,
\end{equation}
where $\Delta> 0$ is the unique solution of the BCS gap equation \eqref{eq:gapeq} (see Proposition \ref{prop.KTDelta.geq.0}) and we denoted $\Delta_\xi := - 2 (2 \pi)^{-d/2} \lambda \hat{V} \star \hat{\xi}$.

\subsubsection{A priori bounds on \texorpdfstring{$\Delta_0$}{Delta0}}

For the proof of \Cref{thm:main}(b) we need some a priori bounds on $\Delta_0$ analogously to those of \Cref{sec.a.priori.Delta.(a)}.
The bounds follow from the following lemma, 
the proof of which is analogous to the argument of \Cref{sec.Birman--Schwinger.(a),sec.a.priori.Delta.(a)} and given in \Cref{sec.proof.lem:Lipschitz}.
\begin{lemma}[{c.f. \cite[Lemma 4]{Hainzl.Seiringer.2008} and \cite[Lemma 13]{Henheik.Lauritsen.2022}}] \label{lem:Lipschitz}
Let $\mathfrak{a}_0 \in H^1(\R^d)$ with $\hat{\mathfrak{a}}_0 > 0$ be the unique $H^1(\R^d)$-normalized ground state of $K_{T_c} + \lambda V$ from Lemma \ref{lem:KTcGS}. Moreover, let $u(p) = (|\Sph^{d-1}|)^{-1/2}$ be the constant function on the sphere $\Sph^{d-1}$ and let 
	\begin{equation} \label{eq:hatphi}
		\hat{\varphi}(p) = 	- \frac{1}{(2\pi)^{d/2}} \int_{\Sph^{d-1}} \hat{V}(p-\sqrt{\mu} q) \, \D \omega(q)\,.
	\end{equation}
	Then $\Delta_0= - 2 (2\pi)^{-d/2} \lambda \hat{V} \star \hat{\mathfrak{a}}_0$ can be expanded as
	\begin{equation} \label{eq:Delta0expand}
		\Delta_0(p) = f(\lambda) [\hat{\varphi}(p) + \lambda \eta_\lambda(p)]
	\end{equation}
	for some positive function $f(\lambda)$ and $\Vert \eta_\lambda \Vert_{L^\infty(\R^d)}$ bounded uniformly in $\lambda > 0$. 
\end{lemma}
After realizing $\hat{\varphi}(\sqrt{\mu}) =  - e_\mu$ by \eqref{eq:emu}, we conclude for small enough $\lambda > 0$ that
\begin{equation*}
	\Delta_0(p) - \Delta_0(\sqrt{\mu}) = \frac{[\hat{\varphi}(p) - \hat{\varphi}(\sqrt{\mu})] + \lambda [ \eta_\lambda(p) - \eta_\lambda(\sqrt{\mu})]}{- e_\mu + \lambda \eta_\lambda(\sqrt{\mu})}\Delta_0(\sqrt{\mu})\,.
\end{equation*}
Now it is an easy computation to see $|\hat{\varphi}(p) - \hat{\varphi}(q)| \le C \min\big\{ ||p|-|q||, 1 \big\}$ for all $p,q \in \R^d$. Thus, 
\begin{equation} \label{eq:Lipschitzbound}
	\big|  \Delta_0(p) - \Delta_0(\sqrt{\mu})  \big| \le C \,  \left( \min\big\{ \big||p|-\sqrt{\mu}\big|, 1 \big\} + \lambda  \right) \, \Delta_0(\sqrt{\mu})\,. 
\end{equation}

\subsubsection{A priori bounds on \texorpdfstring{$\Delta_\xi$}{DeltaXi}}

For the following arguments, we need two estimates on $\Delta_\xi = - 2 (2 \pi)^{-d/2} \lambda \hat{V} \star \hat{\xi}$. 
\begin{itemize}
	\item First, it is a simple consequence of Young's inequality and $\Vert \hat{\xi} \Vert_{L^2}= \Vert {\xi} \Vert_{L^2} = O\big(h^2 \ee^{c/\lambda}\big)$, that 
	\begin{equation} \label{eq:Dxi Linf}
		\Vert \Delta_\xi \Vert_{L^\infty} = \Vert V \Vert_{L^2}  \, O\big(h^2 \ee^{c/\lambda}\big)\,.
	\end{equation}
\item Second, we note that $\Delta_\xi(p) - \Delta_\xi(q)$ is (proportional to) the Fourier transform of $ V(x) \big( 1 - \E^{\I (p-q)\cdot x} \big) \xi(x)$, and thus
\begin{equation*}
	\left\Vert V(x) \big( 1 - \E^{\I (p-q)\cdot x} \big) \right\Vert_{L^2}^2 = \int_{\R^d} |V(x)|^2 \left\vert 1 - \E^{\I (p-q)\cdot x} \right\vert^2 \D x \le C |p-q|^2 \int_{\R^d} |V(x)|^2 |x|^2 \D x\,. 
\end{equation*}
Using radiality of $\Delta$ and $\Delta_0$, we conclude the radiality of $\Delta_\xi$ and therefore 
\begin{equation} \label{eq:Dxi Lip}
	|\Delta_\xi(p) - \Delta_\xi(q)| \le \big| |p| - |q| \big| \, \left\Vert |\cdot| V \right\Vert_{L^2} \, O\big(h^2 \ee^{c/\lambda}\big)\,. 
\end{equation}
\end{itemize}
Recall that $\Vert (1 + |\cdot|) V \Vert_{L^2}< \infty$ by assumption.

\subsubsection{Comparing \texorpdfstring{$\Delta(\sqrt{\mu})$ and $\Xi$}{Delta(sqrt(mu)) and Xi}}
We aim at proving
\begin{equation} \label{eq:Xi=Delta}
	\Xi = \Delta(\sqrt{\mu}) \left( 1 + O\big(\lambda + h^2 \ee^{c/\lambda}\big)  \right)\,. 
\end{equation}
In order to see this, we note that clearly $\Xi = \sqrt{(p^2-\mu)^2 + |\Delta(p)|^2} \le \Delta(\sqrt{\mu})$. For the reverse inequality, let $p \in \R^d$ with $|p^2 - \mu | \le \Xi \le \Delta(\sqrt{\mu})$. Then 
\begin{equation*}
	|\Delta(p) - \Delta(\sqrt{\mu})| \le C \, T_c\,  h \big(  1 + o_{\lambda \to 0}(1) + O(h \ee^{c/\lambda})\big) \cdot \big( \big| |p| - \sqrt{\mu}\big| + \lambda \big)
  + C \big| |p| - \sqrt{\mu}\big| \, h^2 \ee^{c/\lambda} 
\end{equation*}
by application of \eqref{eq:Lipschitzbound} and \eqref{eq:Dxi Lip}. Using that $\Delta(\sqrt{\mu}) \sim T_c h$, as a consequence of \eqref{eq:Delta exp} for $h$ small enough (meaning $h \mathrm{e}^{c/\lambda} \ll 1$), we then conclude
\begin{equation*}
	|\Delta(p) - \Delta(\sqrt{\mu})|	\le  C  \, \big(\lambda + h^2 \ee^{c/\lambda}\big)  \,  \Delta(\sqrt{\mu})
\end{equation*}
In combination with the upper bound, this proves \eqref{eq:Xi=Delta}.

\subsubsection{Conclusion: Proof of Theorem \ref{thm:main}(b)}
We evaluate \eqref{eq:Delta exp} at $p = \sqrt{\mu}$, such that we find
\begin{equation*}
	\Xi = \left[C_{\rm univ} \, T_c \, h \, \big( 1 + o_{\lambda \to 0}(1) +O\big(h \ee^{c/\lambda}\big)\big) + O\big(h^2 \ee^{c/\lambda}\big)\right] \cdot \big( 1 + O\big(\lambda + h^2 \ee^{c/\lambda}\big)  \big)
\end{equation*}
with the aid of \eqref{eq:Dxi Linf} and \eqref{eq:Xi=Delta}. Collecting all the error terms leaves us with 
\begin{equation} \label{eq:Xi final}
	\Xi = C_{\rm univ} \, T_c \, h \big( 1 + o_{\lambda \to 0}(1) + O\big(h \ee^{c/\lambda}\big) \big)\,. 
\end{equation}
Hence, using $\fBCS(h) = C_{\rm univ} h + O(h^2)$, by Lemma \ref{lem:fBCS}, we arrive at Theorem~\ref{thm:main}(b). \qed
\subsubsection{Proof of Proposition \ref{prop:GLmin}}

In the following estimates, we use the shorthand notations (recall the definition of the auxiliary function $g_1$ from \eqref{eq:g01})
\begin{equation} \label{eq:f4f2}
f_4(p) := \frac{g_1((p^2 - \mu)/T_c)}{(p^2 - \mu)/T_c} \qquad f_2(p):= \frac{1}{\cosh^2\left( (p^2 - \mu) / (2 T_c)\right)}\,,
\end{equation}
such that the absolute value of the minimizer $\psi_{\rm GL}$ of \eqref{eq:GLfunctional} is given by
\begin{equation*}
|\psi_{\rm GL}| = T_c \left( \frac{\int_{\R^d} f_2(p)  |K_{T_c}(p)|^2 |\hat{\mathfrak{a}}_0(p)|^2  \D p}{4 \int_{\R^d} f_4(p) |K_{T_c}(p)|^4 |\hat{\mathfrak{a}}_0(p)|^4  \D p} \right)^{1/2} = 
T_c \left( \frac{\int_{\R^d} f_2(p)  |\Delta_0(p)|^2  \D p}{\int_{\R^d} f_4(p) |\Delta_0(p)|^4  \D p} \right)^{1/2}\,. 
\end{equation*}
We denoted $\Delta_0= - 2 (2\pi)^{-d/2} \lambda \hat{V} \star \hat{\mathfrak{a}}_0$ (as in Proposition \ref{prop:GLmin}) and used that $\mathfrak{a}_0 \in \ker\big(K_{T_c} + \lambda V\big)$. Note that, $\Delta_0 = |\Delta_0|$ by means of Proposition \ref{prop.KTDelta.geq.0} and $\hat{V} \le 0$ from Assumption \ref{ass:hatV neg}. 

Next, we add and subtract $|\Delta_0(\sqrt{\mu})|^2$ (resp.~$|\Delta_0(\sqrt{\mu})|^4$) in the integral in the numerator (resp.~denominator). The terms involving $|\Delta_0(\sqrt{\mu})|^j$ are evaluated as follows. 
\begin{lem}[Emergence of $C_{\rm univ}$ in GL theory] \label{lem:Cuniv}
	Let $\mu > 0$. In the limit $T_c/\mu  \to 0$ we have that (recall $C_{\rm univ}$ from \eqref{eq:Cuniv})
	\begin{equation} \label{eq:ratiotoCuniv}
			\left( \int_{\R^d} \frac{1}{\cosh^2\left(\frac{p^2 - \mu}{2 T_c}\right)} \D p \middle/ \int_{\R^d} \frac{g_1((p^2 - \mu)/T_c)}{(p^2 - \mu)/T_c} \D p\right)^{1/2} \longrightarrow \ C_{\rm univ}
	\end{equation}
	for all $d = 1,2,3$. 
\end{lem}
\begin{proof}
	Since the integrands are both radial, we switch to spherical coordinates and neglect the common $|\mathbb{S}^{d-1}|$--factor in numerator and denominator. By splitting the remaining radial integration according to $p^2 \le \mu$ and $p^2 \ge \mu$ and changing the integration variables from $(p^2 - \mu)/2T_c$ to $-t$ resp.~$t$ we find the numerator of \eqref{eq:ratiotoCuniv} being equal to
	\begin{equation} \label{eq:numer}
		2 T_c\mu^{(d-2)/2} \left[\int_{0}^{\mu/2T_c} \hspace{-2mm}\big(1 - 2 \tfrac{T_c}{\mu} t\big)^{(d-2)/2}  +  \int_{0}^{\infty} \big(1 + 2 \tfrac{T_c}{\mu} t\big)^{(d-2)/2} \right] \left(\frac{1}{\cosh^2(t)}\right) \D t\,. 
	\end{equation}
	Similarly, we find the denominator of \eqref{eq:ratiotoCuniv} to equal 
	\begin{equation} \label{eq:denom}
		\frac{2 T_c\mu^{(d-2)/2}}{8} \left[\int_{0}^{\mu/2T_c} \hspace{-2mm}\big(1 - 2 \tfrac{T_c}{\mu} t\big)^{(d-2)/2}  +  \int_{0}^{\infty} \big(1 + 2 \tfrac{T_c}{\mu} t\big)^{(d-2)/2} \right]   \left(\frac{\tanh(t)}{t^3} - \frac{1}{t^2 \cosh^2(t)}\right) \D t\,. 
	\end{equation}
	We now take the ratio of \eqref{eq:numer} and \eqref{eq:denom} and send $T_c/\mu \to 0$. By means of the dominated convergence theorem (note that the integrand in \eqref{eq:denom} behaves as $t^{-3}$ for large $t$) we thus find the limit being given as the ratio of
	\begin{equation*}
		\int_{0}^{\infty} \frac{1}{\cosh^2(t)} \D t \qquad \text{and} \qquad \frac{1}{8}\int_{0}^{\infty} \left(\frac{\tanh(t)}{t^3} - \frac{1}{t^2 \cosh^2(t)}\right) \D t\,.
	\end{equation*}
	While the former is easily evaluated as being equal to one, the latter is given by $\frac{7 \zeta(3)}{8 \pi^2}$ (see, e.g., \cite[3.333.3]{Gradshteyn.Ryzhik.2007}). 
	This proves the claim. 
\end{proof}

With the aid of \eqref{eq:Lipschitzbound} and noting $f_j >0 $, the resulting differences (from adding and subtracting $|\Delta_0(\sqrt{\mu})|^j$) can be estimated as
\begin{equation*}
\left\vert\int_{\R^d} f_j(p)  \big(|\Delta_0(p)|^j -|\Delta_0(\sqrt{\mu})|^j\big)  \D p \right\vert \, \le \, C \, |\Delta_0(\sqrt{\mu})|^j \int_{\R^d} f_j(p) \left( \min\big( \big||p|-\sqrt{\mu}\big|, 1 \big) + \lambda\right) \D p
\end{equation*}
for $j = 2, 4$. These integrals can be treated analogously to \eqref{eq:numer} and \eqref{eq:denom} in Lemma \ref{lem:Cuniv} (for the $\big\vert |p| - \sqrt{\mu}\big\vert$-term, note that $f_j$ essentially concentrates around $|p| \approx \sqrt{\mu}$) and we find them to be smaller than the corresponding leading term $ \int_{\R^d} f_j(p) |\Delta_0(\sqrt{\mu})|^j \D p$ in the limit $\lambda \to 0$ (and hence $T_c \to 0$). Therefore, 
\begin{equation*}
	|\psi_{\rm GL}|  = 
\frac{T_c}{\Delta_0(\sqrt{\mu})} \left( \frac{\int_{\R^d} f_2(p)    \D p}{\int_{\R^d} f_4(p)  \D p} \right)^{1/2} \cdot\big(1 + o_{\lambda \to 0}(1)\big) = C_{\rm univ}	\,  \frac{T_c}{\Delta_0(\sqrt{\mu})} \cdot \big(1 + o_{\lambda \to 0}(1)\big) \,,
\end{equation*}
where we used Lemma \ref{lem:Cuniv} in the last step. As the GL functional \eqref{eq:GLfunctional} is entirely independent of the relative difference to the critical temperature $(T_c-T)/T_c = h^2$, it is clear that all the errors here hold uniform in the parameter $h$. This finishes the proof of Proposition~\ref{prop:GLmin}. \qed

\subsection{Pure angular momentum for \texorpdfstring{$d=2$}{d=2}: Proof of Theorem~\ref{thm:main2}} \label{subsec:proof2d angmom}

\subsubsection{Part (a)} \label{subsubsec:proof2d angmom parta}
The proof of \Cref{thm:main2}~(a) is mostly the same as that of \Cref{thm:main}~(a). 
We sketch the argument here, highlighting the few differences.

\paragraph{The operator \texorpdfstring{$\mcV_\mu$}{Vmu}.}
Using the Birman--Schwinger principle on the operator $K_{T_c} + \lambda V$ (as is done in \cite{frank.hainzl.naboko.seiringer,Hainzl.Seiringer.2008})
we find that, for sufficiently small $\lambda$, the lowest eigenvalue $e_\mu$ of $\mcV_\mu$ (recall \eqref{eq:Vmu}) is an eigenvalue for 
angular momentum $\ell_0$, since this is the angular momentum of the ground state(s) of $K_{T_c} + \lambda V$ by assumption.
Further, since $V$ is radial, the eigenfunctions of $\mcV_\mu$ all have a definite angular momentum.
In particular the first excited state has some angular momentum $\ell_1 \ne \ell_0$:
\begin{equation*}
e_\mu^{(1)} = \inf_{u\perp u_{\pm\ell_0}} \longip{u}{\mcV_\mu}{u} = \longip{u_{\pm\ell_1}}{\mcV_\mu}{u_{\pm\ell_1}},
\end{equation*}
with $u_{\pm\ell}(p) = (2\pi)^{-1/2}e^{\pm i\ell\varphi}$ the eigenfunctions of angular momentum $\ell$.
Here $\varphi$ denotes the angle of $p\in \R^2$ in polar coordinates.
Note that $e_\mu^{(1)}\leq 0$ since $\mcV_\mu$ is a compact operator on an infinite-dimensional space.

\paragraph{A priori spectral information.}
It is proved in \cite[Theorem 2.1]{Deuchert.Geisinger.ea.2018} that there exists a temperature $\tilde T$ such that for temperatures $\tilde T < T < T_c$
the minimizers of the BCS functional are given by 
\begin{equation*}
\hat \alpha_{\pm}(p) = e^{\pm i \ell_0 \varphi} \hat \alpha_0(p)
\end{equation*}
where $\varphi$ denotes the angle of $p\in \R^2$ in polar coordinates, $\hat\alpha_0$ is a radial function,
and $\ell_0$ is the angular momentum given by \Cref{eq:pureang}.
The BCS gap functions are then $\Delta_{\pm}(p) = \Delta_0(p) e^{\pm i \ell_0 \varphi}$, with $\Delta_0$ a radial function.\footnote{This should not be confused with the function $\Delta_0$ used in Section \ref{subsec:GLproof}.}
Further, we have $K_T^{\Delta_0} + \lambda V \geq 0$ for temperatures $T \in (\tilde T, T_c)$ \cite[Proposition 4.3]{Deuchert.Geisinger.ea.2018} and $\ker(K_T^{\Delta_0} + \lambda V) = \mathspan\{\alpha_+, \alpha_-\}$.

\paragraph{The temperature $\tilde T$.}
As discussed in \cite[Remark 2.6]{Deuchert.Geisinger.ea.2018} the temperature $\tilde T$ is given by $\tilde T = T_c(\ell_1)$, 
the critical temperature when restricted to angular momentum $\ell_1$.
Following the argument in \cite{frank.hainzl.naboko.seiringer} (see also \cite[Theorem 1]{hainzl.seiringer.2010}) we find 
\begin{equation*}
\tilde T \leq \begin{cases}
C e^{1/\lambda e_{\mu}^{(1)}} & e_{\mu}^{(1)} < 0, \\
C e^{-c/\lambda^2} & e_\mu^{(1)}=0.
\end{cases}
\end{equation*}
Recalling that $T_c \sim e^{1/\lambda e_\mu}$ and that $e_\mu < e_\mu^{(1)} \leq 0$ then clearly $\tilde T / T_c \leq C e^{-\mathfrak{c}/\lambda}$
for some $\mathfrak{c} > 0$.

\paragraph{Weak a priori bound on \texorpdfstring{$\Delta_\pm$}{Delta}.}
Exactly as in \Cref{eqn.weak.Delta.infty.a.priori} we have $\norm{\Delta_{\pm}}_{L^\infty}\leq C \lambda$.

\paragraph{Birman--Schwinger principle.}
Analogously to \Cref{sec.Birman--Schwinger.(a)} we have by the Birman--Schwinger principle that 
\begin{equation*}
B_{T,\Delta_0} = V^{1/2} \frac{1}{K_T^{\Delta_0}}|V|^{1/2}
\end{equation*}
has $-1$ as its lowest eigenvalue, only the eigenspace is spanned by the two vectors $\phi_\pm = V^{1/2}\alpha_\pm$.
By a completely analogous argument is in \Cref{sec.Birman--Schwinger.(a)} 
we find that 
\begin{equation*}
S_{T,\Delta_0} = \lambda m(T,\Delta_0) \mfF{}_\mu |V|^{1/2} \frac{1}{1 + \lambda V^{1/2} M_{T,\Delta_0}|V|^{1/2}} V^{1/2} \mfF{}_\mu^\dagger 
\end{equation*}
has $-1$ as its lowest eigenvalue with corresponding eigenspace spanned by $\mfF{}_\mu V \alpha_{\pm}$.

\paragraph{A priori bounds on \texorpdfstring{$\Delta_\pm$}{Delta}.}
Analogously to \Cref{lem.Delta.a.priori.bdd} we claim 
\begin{lemma} \label{lem:aprioripureell}
The functions $\Delta_{\pm}$ satisfy the bounds (with slight abuse of notation, recall that $\Delta_0$ is a radial function)
\begin{equation*}
\abs{\Delta_{\pm}(p)} \leq C |\Delta_0(\sqrt\mu)|,
\qquad 
\abs{\Delta_{\pm}(p) - \Delta_{\pm}(q)} \leq C |\Delta_0(\sqrt\mu)| |p-q|\,.
\end{equation*}
\end{lemma} 
\begin{proof}
The proof is analogous to that of \Cref{lem.Delta.a.priori.bdd}.
First we note that $\mcV_\mu= \mfF_\mu V \mfF_\mu^\dagger$ has eigenfunctions of lowest eigenvalue $u_{\pm \ell_0}(p) = (2\pi)^{-1/2}e^{\pm i \ell_0 \varphi}$
and that the operator $S_{T,\Delta_0}$ preserves the angular momentum.
Analogously to the proof of \Cref{lem.Delta.a.priori.bdd} we find 
\begin{equation*}
\Delta_\pm(p) = f_\pm(\lambda) \left(\int_{\S^1} \hat V(p - \sqrt\mu q) u_{\pm\ell_0}(q) \ud \omega(q) + \lambda \eta_{\lambda \pm}(p)\right)
\end{equation*}
with $\norm{\eta_{\lambda\pm}}_{L^\infty}\leq C$ uniformly in $\lambda$.
Evaluating on the Fermi surface $\{p^2 = \mu\}$ we get (recall that $\inf\spec \mcV_\mu=e_\mu$) 
\begin{equation*}
\Delta_{\pm}(\sqrt{\mu}p/|p|) = f_{\pm}(\lambda) \left(e_\mu u_{\pm\ell_0}(p/|p|) + \lambda \eta_{\lambda \pm}(\sqrt{\mu}p/|p|)\right).
\end{equation*}
In particular, we conclude that $\abs{\Delta_0(\sqrt\mu)} = \abs{\Delta_{\pm}(\sqrt{\mu}p/|p|)} > 0$ for $\lambda$ small enough and that $\abs{f_{\pm}(\lambda)}\leq C \abs{\Delta_0(\sqrt\mu)}$.
We conclude the rest of the proof exactly as for \Cref{lem.Delta.a.priori.bdd}.
\end{proof}

The remaining parts of the argument 
(first order analysis of $m$, the exponential vanishing of $\Delta_\pm$, infinite order analysis of $m$, calculation of the integral $m(T,\Delta_0) - m(T_c,0)$
and comparing $\Delta_\pm$ on the Fermi surface with $\Xi$) 
are exactly as in \Cref{sec.first.order.(a),sec.infinite.order.(a),sec.exp.vanish.Delta,sec.calc.int.m-m,sec.compare.Xi.Delta.(a)}
only replacing $\Delta$ and $u$ by $\Delta_\pm$ and $u_{\pm\ell_0}$, respectively. 
This concludes the proof of Theorem~\ref{thm:main2}~(a).

\subsubsection{Part (b)} \label{subsubsec:partb}Again, we highlight only the main differences compared to the proof of Theorem~\ref{thm:main}~(b). 

\paragraph{Ginzburg-Landau functional} Since every function $\hat{\mathfrak{a}}_0$ in kernel of $K_{T_c} + \lambda V$ can be written (in polar coordinates) as
\begin{equation*}
\hat{\mathfrak{a}}_0(p,\varphi) = \hat{\rho}(p) \big[\psi_+ \mathrm{e}^{\mathrm{i} \ell_0 \varphi} + \psi_-  \mathrm{e}^{-\mathrm{i} \ell_0 \varphi}\big]
\end{equation*}
for an appropriate normalized $\hat{\rho} \in L^2((0,\infty); p \ud p)$ and $\psi_{\pm} \in \C$ by Assumption \ref{ass:pureell} (cf.~\eqref{eq:pureang}), the analog of the Ginzburg-Landau functional \eqref{eq:GLfunctional} becomes \cite[Theorems 2.10 and 3.5]{Frank.Lemm.2016}
\begin{equation} \label{eq:GLfunc2}
	\begin{split}
\mathcal{E}_{\rm GL}(\psi_+, \psi_-) = & \left[|\psi_+|^4 + |\psi_-|^4 + 4 |\psi_+|^2 |\psi_-|^2\right] \times \left[\frac{2\pi }{T_c^3}\int_{0}^{\infty} \frac{g_1((p^2 - \mu)/T_c)}{(p^2 - \mu)/T_c} |K_{T_c}(p)|^4 |\hat{\rho}(p)|^4 \, p \, \D p\right]  \\[1mm] 
& \hspace{5mm}- \left[|\psi_+|^2 + |\psi_-|^2\right] \times \left[\frac{\pi }{ T_c}\int_{0}^\infty \frac{1}{\cosh^2\left( (p^2 - \mu) / (2 T_c)\right)} |K_{T_c}(p)|^2 |\hat{\rho}(p)|^2\, p \, \D p \right] \,.
	\end{split}
\end{equation}

\paragraph{Minimizers of the GL functional}
In contrast to \eqref{eq:GLfunctional}, the functional \eqref{eq:GLfunc2} now has \emph{two} (up to a phase unique) minimizers. This follows from observing that $\mathcal{E}_{\rm GL}(\psi_+, \psi_-) = \mathcal{E}_{\rm GL}(\psi_-, \psi_+)$ and that one of the $\psi_{\pm}$ is necessarily zero for any minimizer of \eqref{eq:GLfunc2}. In fact, these minimizers are
\begin{equation*}
(|\psi_{\rm GL}|, 0) \quad \text{and} \quad (0, |\psi_{\rm GL}|) 
\end{equation*}
with $|\psi_{\rm GL}|$ given in \eqref{eq:psiGL} but with $\Delta_0:= - 2 (2\pi)^{-d/2} \lambda \hat{V} \star \hat{\rho}$, where $\hat{\rho}$ is understood as a radial function in $L^2(\R^2)$.\footnote{The fact that $\Delta_0(\sqrt{\mu}) \neq 0$ can be seen in a similar way as in the proof of Lemma \ref{lem:aprioripureell}.} Hence, using the notation from Proposition \ref{prop:GLtheory} (see also \cite[Theorem 2.10]{Frank.Lemm.2016}, which provides a general analog of \eqref{eq:approxmin}--\eqref{eq:approx minimum}, valid also for the concrete functional \eqref{eq:GLfunc2}) and \eqref{eq:Delta exp}, we find that (up to a constant phase) \emph{every} non-zero solution of the BCS gap equation \eqref{eq:gapeq} can be written as 
\begin{equation*}
	\Delta_\pm(p, \varphi) = C_{\rm univ} \, T_c \, h  \, \frac{\Delta_0(p)}{\Delta_0(\sqrt{\mu})} \mathrm{e}^{\pm \mathrm{i} \ell_0 \varphi}\big( 1 + o_{\lambda \to 0}(1) + O\big(h \ee^{c/\lambda}\big)\big)+ \Delta_\xi(p, \varphi)\,. 
\end{equation*}

The rest of the argument (a priori bounds on $\Delta_0(p) \mathrm{e}^{\pm \mathrm{i} \ell_0 \varphi}$ and $\Delta_\xi$, comparison of $|\Delta_\pm|$ on the Fermi surface with $\Xi$) works completely analogously to Section \ref{subsec:GLproof} with similar adjustments as explained in Section \ref{subsubsec:proof2d angmom parta}. This concludes the proof of Theorem \ref{thm:main2}~(b). \qed

\bigskip

\noindent \textit{Acknowledgments.}
We thank Andreas Deuchert, Christian Hainzl, Edwin Langmann, Marius Lemm, Robert Seiringer, and Jan Philip Solovej for helpful discussions, and Edwin Langmann and Robert Seiringer for valuable comments on an earlier version of the manuscript.
J.H.~gratefully acknowledges partial financial support by the ERC Advanced Grant ``RMTBeyond'' No.~101020331.
A.B.L. gratefully acknowledges partial financial support by the Austrian Science Fund (FWF) through project number~I~6427-N (as part of the SFB/TRR~352).

\appendix
\section{Additional proofs} \label{app:auxandadd}

\subsection{Uniqueness of the minimizer: Proof of Proposition \ref{prop.KTDelta.geq.0}} \label{app:Deltaunique}

Finally, we present the proof of Proposition \ref{prop.KTDelta.geq.0}.
\begin{proof}[Proof of Proposition \ref{prop.KTDelta.geq.0}] 
We remark that the argument has already partly been sketched in \cite{hainzl.seiringer.16, Deuchert.Geisinger.ea.2018}. 
The key observation for our proof is, that, if $\hat{V} \le 0$, then
	\begin{equation} \label{eq:keypoint}
		\braket{\hat{\alpha} \vert \hat{V} \star \hat{\alpha}} \ge \braket{\abs{\hat{\alpha}} \vert \hat{V} \star \abs{\hat{\alpha}}}\,. 
	\end{equation} 
	Let $(\hat{\gamma}, \hat{\alpha})$ minimize the BCS functional \eqref{eq:functional}. Then, by means of \eqref{eq:keypoint}, we have $\mathcal{F}_T[(\hat{\gamma}, \hat{\alpha})] \ge \mathcal{F}_T[(\hat{\gamma}, \abs{\hat{\alpha}}) ]$, hence also $(\hat{\gamma}, \abs{\hat{\alpha}})$ is a minimizer. Consequently, the (inverse) Fourier transform of $\abs{\hat{\alpha}}$ is an eigenvector of $K_T^\Delta + V$ with $\Delta= - 2 (2 \pi)^{-d/2}\hat{V} \star \abs{\hat{\alpha}}$ to the eigenvalue zero. Note that, using continuity of $\Delta$ and the BCS gap equation \eqref{eq:gapeq}, we not only have $\abs{\hat{\alpha}} \ge 0$ but also $\abs{\hat{\alpha}} > 0$ everywhere (see \cite[Lemma 2.1]{Hainzl.Seiringer.2008}). By the observation \eqref{eq:keypoint} again, we find that for any ground state $\hat{\alpha}_{\rm GS}$ of $K_T^\Delta + V$ also $\abs{\hat{\alpha}_{\rm GS}}$ is a ground state. But $\abs{\hat{\alpha}_{\rm GS}}$ is non--orthogonal to $\abs{\hat{\alpha}}$, which implies that zero has to be the lowest eigenvalue of $K_T^\Delta + V$, i.e.
	\begin{equation} \label{eq:ge0}
		K_T^\Delta + V \ge 0\,.
	\end{equation}
	By writing out \eqref{eq:keypoint}, we see that the inequality is actually an application of Cauchy-Schwarz and thus becomes strict, unless $\hat{\alpha}(p) = \E^{\I \phi} \abs{\hat{\alpha}(p)}$ for some fixed $\phi \in \R$. Therefore, by repeating the above arguments, we find that the ground state of \eqref{eq:ge0} is non--degenerate and we have proven item (i). 
	
	In order to prove item (ii), let $\Gamma_i \equiv (\gamma_i, \alpha_i)$, $i = 1,2$, be two (non--trivial) minimizers of the BCS functional \eqref{eq:functional} and denote the corresponding gap functions by $\Delta_1$ resp.~$\Delta_2$. We now apply the relative entropy identity (see \cite{Frank.Hainzl.ea.2012} and \cite[Prop.~5.2]{Frank.Lemm.2016}) and a simple trace inequality (see \cite[Lemma~3]{Frank.Hainzl.ea.2012} and \cite[Lemma~5.7]{Frank.Lemm.2016}) to find that 
	\begin{equation*}
		\mathcal{F}_T[\Gamma_1] - \mathcal{F}_T[\Gamma_2] \ge \longip{(\alpha_1 - \alpha_2)}{K_T^{\Delta_1} + V}{(\alpha_1 - \alpha_2)} \ge 0
	\end{equation*}
	(and the same inequality with indices $1$ and $2$ interchanged) by means of \eqref{eq:ge0}. Since $\mathcal{F}_T[\Gamma_1] = \mathcal{F}_T[\Gamma_2]$, this implies $(\alpha_1 - \alpha_2) \in \ker(K_T^{\Delta_1} + V)$ and thus $\alpha_2 = \psi_{21} \alpha_1$ for some $\psi_{21} \in \C\setminus \{0\}$ (recall from (i) that $\ker(K_T^{\Delta_1} + V)$ is one--dimensional). From this we conclude 
	\begin{equation*}
		\big(K_T^{\psi_{21}\Delta_1} + V\big) \alpha_1 = 0\,.
	\end{equation*}
	Now, the pointwise strict monotonicity of $\abs{\psi_{21}}\mapsto K_T^{\psi_{21}\Delta_1}(p)$ together with the fact that one can choose  
	$\abs{\hat{\alpha}_1}$ to be strictly positive, implies that $|\psi_{21}| = 1$ and we have shown uniqueness of minimizers up to a constant phase, which can be chosen in such a way that it ensures strict positivity of $\hat{\alpha}$. 
	Finally, it is shown in \cite[Proposition~2.9]{Deuchert.Geisinger.ea.2018} that if $\alpha$ is not radial, then \eqref{eq:ge0} is violated. Radiality of the corresponding $\gamma$ follows from \eqref{eq:ELeqgamma}. This finishes the proof. 
\end{proof}

\subsection{Proofs of technical lemmas within the proof of Theorem \ref{thm:main}}\label{sec.proofs.technical.lemmas}
This section contains the proofs of Lemmas \ref{lem.bdd.VMV}, \ref{lem.bdd.V.MTDelta.minus.MTc.V}, and \ref{lem:Lipschitz}.

\subsubsection{Proof of \texorpdfstring{\Cref{lem.bdd.VMV}}{Lemma~\ref*{lem.bdd.VMV}}}\label{sec.proof.lem.bdd.VMV}

The argument is slightly different in dimensions $d=1,2,3$. 
The case $d=3$ is similar to \cite{Hainzl.Seiringer.2008,frank.hainzl.naboko.seiringer} 
and the case $d=1,2$ is similar to \cite{Henheik.Lauritsen.ea.2023}.
\paragraph*{The case \texorpdfstring{$d=3$:}{d=3:}}
We write 
\begin{equation}
\label{eqn.decompose.VMV}
  V^{1/2} M_{T,\Delta} |V|^{1/2}
  = V^{1/2} \frac{1}{p^2}|V|^{1/2}
   + V^{1/2} \left(M_{T,\Delta} - \frac{1}{p^2} \chi_{|p| > \sqrt{2\mu}}\right) |V|^{1/2}
   - V^{1/2} \frac{1}{p^2}\chi_{|p| \leq \sqrt{2\mu}} |V|^{1/2}.
\end{equation}
The first term in \Cref{eqn.decompose.VMV} has kernel (proportional to)
\begin{equation}
\label{eqn.kernel.Vp^2V}
  V(x)^{1/2} \frac{1}{|x-y|} |V(y)|^{1/2} \in L^2(\R^3 \times \R^3)
\end{equation}
by the Hardy--Littlewood--Sobolev inequality \cite[Theorem 4.3]{analysis}.
The kernel of the second term in \Cref{eqn.decompose.VMV} is given by 
\begin{equation}\label{eqn.kernel.V(M-p^2)V}
\begin{aligned}
  V(x)^{1/2} |V(y)|^{1/2} 
  \frac{1}{(2\pi)^3} 
  & \Biggl[
  \int_{|p|< \sqrt{2\mu}} \frac{1}{K_T^\Delta(p)} \left( e^{ip(x-y)} - e^{i\sqrt{\mu}p/|p| (x-y)} \right) \ud p
  \\ & \hphantom{\Biggl[} \quad 
  +
  \int_{|p| > \sqrt{2\mu}} \left(\frac{1}{K_T^\Delta(p)} - \frac{1}{p^2}\right) e^{ip(x-y)} \ud p
  \Biggr].
\end{aligned}
\end{equation}
We compute the angular integral first. 
In the first term the integral is 
\begin{align} \label{eq:app1}
4\pi \int_{0}^{\sqrt{2\mu}} \frac{1}{K_T^\Delta(p)} \left[ \frac{\sin |p| |x-y|}{|p| |x-y|} - \frac{\sin \sqrt{\mu} |x-y|}{\sqrt{\mu} |x-y|}\right] |p|^2 \ud |p|.
\end{align}
Here we bound $\abs{ \frac{\sin a }{a} - \frac{\sin b }{b} } \leq C \frac{|a-b|}{a+b}$ for $a,b > 0$.
Thus 
we get the bound 
\begin{equation*}
\abs{\eqref{eq:app1}} \leq 
  C \int_0^{\sqrt{2\mu}} \frac{1}{K_T^\Delta(p)} \frac{\abs{p - \sqrt{\mu}}}{p + \sqrt{\mu}} p^2 \ud p
  \leq C \int_0^{\sqrt{2\mu}} \frac{1}{|p^2-\mu|} \frac{\abs{p - \sqrt{\mu}}}{p + \sqrt{\mu}} p^2 \ud p
  \leq C .  
\end{equation*}
In the second term the integral is (bounded by)
\begin{align*}
  4\pi \int_{\sqrt{2\mu}}^\infty \abs{\frac{1}{K_T^\Delta(p)} - \frac{1}{p^2}} \frac{|\sin |p| |x-y||}{|p| |x-y|} |p|^2 \ud |p|
  \leq \frac{4\pi}{|x-y|}\int_{\sqrt{2\mu}}^\infty \abs{\frac{1}{K_T^\Delta(p)} - \frac{1}{p^2}} p \ud p.
\end{align*}
To bound the remaining integral we bound 
\begin{equation*}
  \abs{\frac{1}{K_T^\Delta(p)} - \frac{1}{p^2}}
  \leq \frac{\abs{\tanh \frac{\sqrt{|p^2-\mu|^2 + \Delta(p)^2}}{2T} - 1}}{\sqrt{|p^2-\mu|^2 + \Delta(p)^2}} 
  +
  \abs{\frac{1}{\sqrt{|p^2-\mu|^2 + \Delta(p)^2}} - \frac{1}{|p^2-\mu|}}
  +
  \abs{\frac{1}{|p^2-\mu|} - \frac{1}{p^2}}.
\end{equation*}
Note first that $\abs{\tanh x - 1} \leq 2 e^{-2x}$. Thus, we have 
\begin{align*}
\int_{\sqrt{2\mu}}^\infty 
\frac{\abs{\tanh \frac{\sqrt{|p^2-\mu|^2 + \Delta(p)^2}}{2T} - 1}}{\sqrt{|p^2-\mu|^2 + \Delta(p)^2}} 
p \ud p
& \leq 2 \int_{\sqrt{2\mu}}^\infty e^{- \sqrt{|p^2-\mu|^2 + \Delta(p)^2}/T} \frac{1}{\sqrt{|p^2-\mu|^2 + \Delta(p)^2}} p \ud p
\\ & 
\leq 2 \int_{\sqrt{2\mu}}^\infty e^{- |p^2-\mu|/T} \frac{p}{|p^2-\mu|}  \ud p
\leq C T .
\end{align*}
Next, we estimate 
\begin{equation} \label{eq:app2}
	\begin{split}
    & \abs{\frac{1}{\sqrt{|p^2-\mu|^2 + \Delta(p)^2}} - \frac{1}{|p^2-\mu|}}
\\ & \quad = \frac{1}{|p^2-\mu|} \frac{\Delta(p)^2}{\sqrt{|p^2-\mu|^2 + \Delta(p)^2}\left(|p^2-\mu| + \sqrt{|p^2-\mu|^2 + \Delta(p)^2}\right)}
\\ & \quad  \leq 
\frac{1}{|p^2-\mu|} \frac{\norm{\Delta}_{L^\infty}^2 }{\sqrt{|p^2-\mu|^2 + \norm{\Delta}_{L^\infty}^2} 
	\left(|p^2-\mu| + \sqrt{|p^2-\mu|^2 + \norm{\Delta}_{L^\infty}^2}\right)}
\leq \frac{\norm{\Delta}_{L^\infty}^2}{2 |p^2-\mu|^3}\,,
	\end{split}
\end{equation}
using pointwise monotonicity in $\Delta(p)$. Thus, changing variables to $u  =p^2-\mu$ we have 
\begin{align*}
\int_{\sqrt{2\mu}}^\infty 
\abs{\frac{1}{\sqrt{|p^2-\mu|^2 + \Delta(p)^2}} - \frac{1}{|p^2-\mu|}}
p \ud p
& \leq \frac{1}{4} \norm{\Delta}_{L^\infty}^2 \int_\mu^\infty \frac{1}{u^3} \ud u \leq C \norm{\Delta}_{L^\infty}^2 . 
\end{align*}
Finally,
\begin{equation*}
  \int_{\sqrt{2\mu}}^\infty \abs{\frac{1}{|p^2-\mu|} - \frac{1}{p^2}} p \ud p \leq C.
\end{equation*}
We conclude that the kernel of the second term in \Cref{eqn.decompose.VMV} is bounded by 
\begin{equation*}
  |V(x)|^{1/2} \left(\frac{1 + T + \norm{\Delta}_{L^\infty}^2}{|x-y|} +1 \right) |V(y)|^{1/2} \in L^2(\R^3\times \R^3).
\end{equation*}
Finally, the last term of \Cref{eqn.decompose.VMV} has kernel 
\begin{equation}\label{eqn.kernel.Vp2chiV}
  4\pi V(x)^{1/2} |V(y)|^{1/2} \int_0^{\sqrt{2\mu}} \frac{\sin p|x-y|}{p|x-y|} \ud p
  \in L^2\left(\R^3\times \R^3\right)
\end{equation}
since the integral is bounded by $\sqrt{2\mu}$.

\paragraph*{The cases \texorpdfstring{$d=1$ and $d=2$:}{d=1 and d=2:}}
The kernel of $M_{T,\Delta}$ is given by 
\begin{equation}
\label{eqn.kernel.MTDelta.d=1.d=2}
  M_{T,\Delta}(x,y)
  = \frac{1}{(2\pi)^d} 
  \left[
  \int_{|p| < \sqrt{2\mu}} \frac{1}{K_T^\Delta} \left(e^{ip(x-y) } - e^{i\sqrt{\mu} p/|p| (x-y)}\right)
  +
  \int_{|p| > \sqrt{2\mu}} \frac{1}{K_T^\Delta} e^{ip(x-y)}
  \right]
\end{equation}
Now, one may bound 
$K_T^\Delta \geq |p^2-\mu|$ uniformly in $T,\Delta$.
Then we may bound $M_{T,\Delta}$ exactly as in \cite[Lemma 3.4]{Henheik.Lauritsen.ea.2023}.
That is, we have the bounds 
\begin{equation*}
  \norm{V^{1/2} M_{T,\Delta} |V|^{1/2}}_{\textnormal{HS}}^2
  \lesssim
  \begin{cases}
\norm{V}_{L^1}^2 + \norm{V}_{L^1} \int_\R |V(x)| \left[1 + \log ( 1 + \sqrt{\mu}|x|)\right]^2 \ud x
  & d=1,
  \\
  \norm{V}_{L^1}^2 + \norm{V}_{L^p}^2
  & d=2
  \end{cases}
\end{equation*}
for any $1 < p \leq 4/3$.
This concludes the proof of \Cref{lem.bdd.VMV}.
\qed

\subsubsection{Proof of \texorpdfstring{\Cref{lem.bdd.V.MTDelta.minus.MTc.V}}{Lemma~\ref*{lem.bdd.V.MTDelta.minus.MTc.V}}}
\label{sec.proof.lem.bdd.V.MTDelta.minus.MTc.V}
To bound the difference, first note that by computing the angular integrals we have 
\begin{equation}
\label{eqn.MTDelta.minus.MTc}
\begin{aligned}
  \left[M_{T,\Delta} - M_{T_c,0}\right](x,y)
    &
    = \frac{|\S^{d-1}|}{(2\pi)^d} 
    \Biggl[
    \int_{0}^{\sqrt{2\mu}} 
      \left(\frac{1}{K_T^\Delta} - \frac{1}{K_{T_c}}\right)
      \left(j_d(p|x-y|) - j_d(\sqrt{\mu}|x-y|)\right)
      p^{d-1}
      \ud p
  \\ & \hphantom{= \frac{|\S^{d-1}|}{(2\pi)^d} \Biggl[}
    +
    \int_{\sqrt{2\mu}}^\infty  
      \left(\frac{1}{K_T^\Delta} - \frac{1}{K_{T_c}}\right)
      j_d(p|x-y|)
      p^{d-1}
      \ud p
    \Biggr]
\end{aligned}
\end{equation}
where
\begin{equation*}
  j_d(x) = \frac{1}{|\S^{d-1}|}\int_{\S^{d-1}} e^{ix \omega} \ud \omega
  =
  \begin{cases}
  \cos x & d=1
  \\
  J_0(|x|)
  & d=2 \\
    \frac{\sin |x|}{|x|} & d=3
  \end{cases}
\end{equation*}
with $J_0$ the zero'th Bessel function.

Bounding \Cref{eqn.MTDelta.minus.MTc} 
is similar in spirit to the proof of \Cref{lem.bdd.VMV} above.
We bound 
\begin{equation*}
  \abs{\frac{1}{K_T^\Delta} - \frac{1}{K_{T_c}}}
  \leq 
  \frac{\abs{\tanh \frac{E_\Delta}{2T} -1}}{E_\Delta}
  +
  \abs{\frac{1}{E_\Delta} - \frac{1}{|p^2-\mu|}}
  +
  \frac{\abs{1 - \tanh \frac{|p^2-\mu|}{2T_c}}}{|p^2-\mu|}.
\end{equation*}
We bound the first term as follows 
\begin{equation*}
  \frac{\abs{\tanh \frac{E_\Delta}{2T} -1}}{E_\Delta}
  \leq 2 e^{-E_\Delta/T} \frac{1}{E_\Delta}
  \leq 2 e^{-|p^2-\mu|/T} \frac{1}{|p^2-\mu|}
  \leq 2 e^{-|p^2-\mu|/T_c} \frac{1}{|p^2-\mu|}.
\end{equation*}
Similarly, 
\begin{equation*}
    \frac{\abs{1 - \tanh \frac{|p^2-\mu|}{2T_c}}}{|p^2-\mu|}
\leq  2 e^{-|p^2-\mu|/T_c} \frac{1}{|p^2-\mu|}.
\end{equation*}
Finally, we estimate, exactly as in \eqref{eq:app2} in the course of proving \Cref{lem.bdd.VMV},
\begin{equation*}
\begin{aligned}
  \abs{\frac{1}{E_\Delta} - \frac{1}{|p^2-\mu|}}
  & \leq 
  \frac{1}{|p^2-\mu|} \frac{\norm{\Delta}_{L^\infty}^2 }{\sqrt{|p^2-\mu|^2 + \norm{\Delta}_{L^\infty}^2} 
  \left(|p^2-\mu| + \sqrt{|p^2-\mu|^2 + \norm{\Delta}_{L^\infty}^2}\right)}
\\ & \leq \frac{\norm{\Delta}_{L^\infty}^2}{2 |p^2-\mu|^3}\,.
\end{aligned}
\end{equation*}
We will use the first bound for the first integral in \Cref{eqn.Delta.bdd.Delta.sqrt.mu} and the second bound for the second integral in \Cref{eqn.Delta.bdd.Delta.sqrt.mu}.
Note further that 
$\frac{1}{\sqrt{x^2+A^2} 
  \left(x + \sqrt{x^2+A^2}\right)}$ is decreasing in $x$
and $|p^2-\mu| \geq \sqrt\mu |p - \sqrt\mu|$.
That is, we have the bound 
\begin{multline}
\label{eqn.bdd.1/K.minus.1/K}
  \abs{\frac{1}{K_T^\Delta} - \frac{1}{K_{T_c}}}
  \leq
  4 e^{-|p^2-\mu|/T_c} \frac{1}{|p^2-\mu|}
  + 
\chi_{|p| > \sqrt{2\mu}}
\frac{\norm{\Delta}_{L^\infty}^2}{2 |p^2-\mu|^3}
  \\
  + 
  \chi_{|p| < \sqrt{2\mu}}
  \frac{1}{\sqrt\mu|p-\sqrt{\mu}|} 
  \frac{\norm{\Delta}_{L^\infty}^2 }{\sqrt{\mu |p-\sqrt{\mu}|^2 + \norm{\Delta}_{L^\infty}^2} 
  \left(\sqrt\mu|p-\sqrt{\mu}| + \sqrt{\mu |p-\sqrt{\mu}|^2 + \norm{\Delta}_{L^\infty}^2}\right)}.
\end{multline}
In the first integral in \Cref{eqn.MTDelta.minus.MTc} we bound $|j_d(a) - j_d(b)| \leq C|a-b|$. 
Then the contribution of the first term of \Cref{eqn.bdd.1/K.minus.1/K} to the first integral in \Cref{eqn.MTDelta.minus.MTc}
is bounded by (changing variables to $s = \sqrt\mu|p-\sqrt\mu|/T_c$)
\begin{align*}
\int_0^{\sqrt{2\mu}} 
e^{-|p^2-\mu|/T_c} \frac{1}{|p^2-\mu|} |p-\sqrt\mu| |x-y| p^{d-1} \ud p
& \leq 
C |x-y|  \int_0^{\sqrt{2\mu}} 
e^{-\sqrt\mu|p-\sqrt\mu|/T_c}  \ud p
\\ 
& \leq 
 T_c  |x-y| \int_0^{\mu/T_c} e^{-s} \ud s
\leq C T_c  |x-y|.
\end{align*}
Next, the contribution of the last term of \Cref{eqn.bdd.1/K.minus.1/K} is bounded by 
(changing variables to $s = \sqrt\mu|p-\sqrt\mu|/\norm{\Delta}_{L^\infty}$)
\begin{align*}
& 
\int_0^{\sqrt{2\mu}} \hspace{-4mm}
\frac{\norm{\Delta}_{L^\infty}^2 }{\sqrt{\mu |p-\sqrt{\mu}|^2 + \norm{\Delta}_{L^\infty}^2} 
  \left(|p-\sqrt{\mu}|\sqrt\mu + \sqrt{\mu |p-\sqrt{\mu}|^2 + \norm{\Delta}_{L^\infty}^2}\right)}
  \frac{|p-\sqrt\mu| |x-y| p^{d-1}}{|p-\sqrt{\mu}|\sqrt\mu}\ud p
\\
& \quad \leq 
C \norm{\Delta}_{L^\infty}  |x-y| \int_0^{\mu/\norm{\Delta}_{L^\infty}} \frac{1}{\sqrt{s^2+1}(s + \sqrt{s^2-1})} \ud s
\leq C \norm{\Delta}_{L^\infty}  |x-y|. 
\end{align*}
Next we estimate the last integral of \Cref{eqn.MTDelta.minus.MTc}.
Here we note that $\abs{j_d(a)}\leq C$.
Then the contributions of the first and second term in \Cref{eqn.bdd.1/K.minus.1/K} to \Cref{eqn.MTDelta.minus.MTc}
is bounded by 
\begin{align*}
\int_{\sqrt{2\mu}}^\infty e^{-|p^2-\mu|/T_c} \frac{1}{|p^2-\mu|} p^{d-1} \ud p
& \leq C \int_{\sqrt{2\mu}}^\infty e^{-|p^2-\mu|/T_c} |p^2-\mu|^{d/2-2} p \ud p
\leq C T_c  e^{-\mu/T_c} 
\leq C T_c 
\\ \text{and} \qquad 
\int_{\sqrt{2\mu}}^\infty \frac{\norm{\Delta}_{L^\infty}^2}{|p^2-\mu|^3} p^{d-1} \ud p
& \leq C \norm{\Delta}_{L^\infty}^{2} 
\,.
\end{align*}
We conclude that (using $\norm{\Delta}_{L^\infty} \leq C T_c$)
\begin{align*}
\norm{V^{1/2} ( M_{T,\Delta} - M_{T_c,0}) |V|^{1/2}}_{\textnormal{HS}}^2
& \leq 
C T_c^2  \iint |V(x)| |V(y)| (\mu |x-y|^2 + 1) \ud x \ud y 
\\ & \leq 
C T_c^2  \left(\mu \norm{|\cdot|^2V}_1 \norm{V}_1 + \norm{V}_1^2\right)
\leq C e^{-c/\lambda} 
\end{align*}
by assumption on $V$. 
This finishes the proof of \Cref{lem.bdd.V.MTDelta.minus.MTc.V}.
\qed

\subsubsection{Proof of \texorpdfstring{\Cref{lem:Lipschitz}}{Lemma~\ref*{lem:Lipschitz}}}\label{sec.proof.lem:Lipschitz}
	The proof is very similar to the ones of \cite[Lemma~4]{Hainzl.Seiringer.2008} and \cite[Lemma~13]{Henheik.Lauritsen.2022}
  and follows from a Birman--Schwinger argument analogously to \Cref{sec.Birman--Schwinger.(a),sec.a.priori.Delta.(a)}.

  First of all, recall from Proposition \ref{prop.KTDelta.geq.0} (ii), that $K_{T_c} + \lambda V$ has $0$ as a (non-degenerate) ground state eigenvalue, which, 
	by the Birman--Schwinger principle, is equivalent to the fact that the Birman-Schwinger operator
$B_{T_c} := \lambda V^{1/2} K_{T_c}^{-1} |V|^{1/2}$
  has $-1$ as its (non-degenerate) ground state eigenvalue. 
  As in \Cref{sec.Birman--Schwinger.(a)}, defining $m(T_c) := m(T_c, 0)$ (recall \eqref{eqn.define.m(T,Delta)}), we decompose $B_{T_c}$ as 
  \begin{equation*}
    B_{T_c} = \lambda m(T_c) V^{1/2} \mathfrak{F}_\mu^\dagger \mathfrak{F}_\mu |V|^{1/2} + \lambda V^{1/2} M_{T_c} |V|^{1/2}\,,
  \end{equation*}
  where $M_{T_c}$ is such that this holds. It has been shown in \cite[Lemma 2]{frank.hainzl.naboko.seiringer} (for $d=3$) and \cite[Lemma~3.4]{Henheik.Lauritsen.ea.2023} (for $d=1,2$), 
  that the Hilbert-Schmidt norm $\Vert V^{1/2} M_{T_c} |V|^{1/2} \Vert_{\rm HS}$ of the second term is uniformly bounded for small $T_c$ (i.e.~small $\lambda$).

  Then, by an argument completely analogous to the one in the proof of Lemma \ref{lem.Delta.a.priori.bdd} in \Cref{sec.a.priori.Delta.(a)}
  we find that 
  $\Delta_0 = f(\lambda) \big[ \hat{\varphi}+ \lambda \eta_\lambda\big]$ with $\hat{\varphi}$ defined in \eqref{eq:hatphi}
  and $\eta_\lambda$ has $\norm{\eta_\lambda}_{L^\infty}\leq C$ uniformly in small $\lambda$ (cf.~\eqref{eq:Delta0expand}).
This concludes the proof of \Cref{lem:Lipschitz}.
\qed

\printbibliography
\end{document}